%% file: condmat-v4.tex
\documentclass[aps,prl,twocolumn,superscriptaddress]{revtex4-2}

\usepackage{amsthm}
\usepackage{amsmath,amssymb,braket}
\usepackage{amsfonts}
\usepackage{graphicx}
\usepackage{color}
\usepackage{bm}
\usepackage[dvipsnames]{xcolor}
\usepackage{graphicx}
\usepackage{hyperref} 
\usepackage{cleveref} 
\newcommand{\sectionprl}[1]{{\par\it #1.---}}

\newtheorem{lemma}{Lemma}
\usepackage{hyperref} 
\usepackage{cleveref}
\usepackage{tikz}
\usepackage{tikz-3dplot}

\usepackage[normalem]{ulem}
\usepackage{xcolor}

\begin{document}

\title{Universal criticality of entropy production in chemical reaction networks}
\author{Kyota Tamano}
\email{tamano.kyota.45w@st.kyoto-u.ac.jp}
\affiliation{Department of Physics, Kyoto University, Kyoto 606-8502, Japan}
\author{Keiji Saito}
\email{keiji.saitoh@scphys.kyoto-u.ac.jp}
\affiliation{Department of Physics, Kyoto University, Kyoto 606-8502, Japan}

\date{\today}

\begin{abstract}
Stochastic thermodynamics gives universal relations for microscopic entropy production, yet its critical behavior at macroscopic nonequilibrium transitions remains unclassified. We study well-mixed reversible chemical reaction networks in the macroscopic-first limit, where transitions arise as local bifurcations of mass-action dynamics. Using linear-noise formulas, center-manifold normal forms, and Floquet theory, we obtain generic exponents for entropy-production fluctuations and responses at pitchfork, transcritical, saddle-node, and Hopf bifurcations. Beyond this low-order classification, a trajectory-space Cram\'{e}r–Rao type bound yields the universal scaling inequality $\alpha - 2 \beta \ge 0$. Hence divergent responses require divergent fluctuations, but not conversely, making entropy-production fluctuations a sharper probe of nonequilibrium criticality.
\end{abstract}

\maketitle

\sectionprl{Introduction}
Equilibrium thermodynamics was established through the study of macroscopic systems with Avogadro-scale particle numbers, providing a solid foundation for investigating many thermodynamic phenomena, such as equilibrium phase transitions \cite{landau1980statistical}. Meanwhile, the past three decades have witnessed crucial developments in nonequilibrium thermodynamics, primarily focusing on small systems subject to large fluctuations \cite{jarzynskinonequilibrium}. Stochastic thermodynamics has established a thermodynamically consistent definition of entropy \cite{seifert2025stochastic}. This leads to many discoveries of universal relations \cite{BaratoPhysRevLett.114.158101TUR, GingrichPhysRevLett.116.120601TUR, dechant2018current, aslyamov2025nonequilibrium, SeifertPhysRevLett.95.040602IntegratedFluctuationTheorem, CrooksPhysRevE.60.2721CrooksRelation, EvansPhysRevLett.71.2401SecondLawViolation, GallavottiPhysRevLett.74.2694, Lebowitz1999GCsymmetry, Jarzynski2000HamiltonianFluctuationTheorem, Gaspard10.1063/1.1688758ChemFT,shiraishi2018speed,dechant2022minimum,van2023thermodynamic,lee2022speed,funo2019speed}.
\begin{table*}[ht]

\begin{ruledtabular}
\begin{tabular}{lcccc}
Bifurcation type
  & Supercritical Pitchfork
  & Transcritical
  & Saddle-Node
  & Supercritical Hopf \\
\hline

Normal form
  & $\dot{c}=(\theta-\theta_c)c-c^3$
  & $\dot{c}=(\theta-\theta_c)c-c^2$
  & $\dot{c}=(\theta-\theta_c)-c^2$
  & $\begin{cases}
       \dot{c}_1=(\theta-\theta_c)c_1-\omega_0 c_2- c_1(c_1^2+c_2^2)\\
       \dot{c}_2=\omega_0 c_1+(\theta-\theta_c)c_2-c_2(c_1^2+c_2^2)
     \end{cases}$ \\

Stable solutions
  & \raisebox{-0.5\height}{\scalebox{0.7}{\input{./pitchfork.tex}}}
  & \raisebox{-0.5\height}{\scalebox{0.7}{\input{./transcritical.tex}}}
  & \raisebox{-0.5\height}{\scalebox{0.7}{\input{./saddlenode.tex}}}
  & \raisebox{-0.5\height}{\scalebox{0.7}{\input{./hopf.tex}}}
\\
\hline
\hline
Schematic of ${\rm Var}\sigma$
  & \raisebox{-0.5\height}{\scalebox{0.7}{\input{./div_twosided.tex}}}
  & \raisebox{-0.5\height}{\scalebox{0.7}{\input{./div_twosided.tex}}}
  & \raisebox{-0.5\height}{\scalebox{0.7}{\input{./div_saddle_node.tex}}}
  & \raisebox{-0.5\height}{\scalebox{0.7}{\input{./div_lambdatype.tex}}} 
\\
$(\alpha_{-}, \alpha_{+})$
  & (\,2\,~, 2\,) & (\,2~, 2~) & ($\varnothing$, 1) &  ($0_-$, ~1\,) \\
$(\beta_{-}, \beta_{+})$
  & ($0_-$, ${1\over 2}$) & ($0_-$,$0_-$) & ($\varnothing$, ${1\over 2}$) & ($0_-$, $0_-$) 
%
\end{tabular}
\end{ruledtabular}
\caption{\label{tab:critical-exponent}Summary of the universality classes in the entropy production. The stable solutions for the normal form equations are depicted with a blue curve. The exponents $\alpha_{\pm}$ and $\beta_{\pm}$ are the diverging exponents for $\theta > \theta_c$ and $\theta <\theta_c$, respectively. The symbol $0_-$ means that the exponents are nonpositive meaning no divergence. $\varnothing$ implies that the exponent depends on specific models. These exponents are the leading generic exponents; symmetry-protected or fine-tuned cases may show reduced singularities. The pitchfork and Hopf bifurcation show asymmetric behavior in the critical behavior in $\partial_{\theta}\sigma$ and ${\rm Var}\sigma$, respectively. One can also check that the generic relation $\alpha - 2\beta \ge 0$ valid for any types of bifurcations (See Eq.(\ref{cramerraoineq})) is satisfied. }
\end{table*}

The robust microscopic framework of stochastic thermodynamics now facilitates the exploration of universal entropy properties in macroscopic scale, opening new avenues such as the large-deviation analysis of entropy distributions \cite{FalascoRevModPhys.97.015002} and the Maxwell demon in macro-scale \cite{freitas2022maxwell, Parrondo10.1063/1.1388006MaxwellSymmetry}. Notably, recent numerical studies have revealed a nontrivial divergence of entropy production fluctuations at nonequilibrium phase transitions including chemical reaction networks and nonequilibrium spin systems \cite{NguyenPhysRevE.102.022101Schlogl, Remlein10.1063/5.0203659, Nguyen10.1063/1.5032104ChemicalOscillator, FiorePhysRevE.104.064123DiscontinuousPhaseTransition, ShimPhysRevE.93.012113ActiveMatterBased, PtaszynskiPhysRevE.111.034125CurieWeiss, OberreiterPhysRevLett.126.020603TimeCristal, chudak2025synchronizationthermodynamicallyconsistentstochastic, KewmingPhysRevA.106.033707Kerr}. To develop this direction further, it is instructive to revisit the well-established framework of equilibrium phase transitions. Spin systems have long provided a prototypical paradigm for equilibrium critical phenomena, where the universality classes of critical exponents are classified based on the underlying symmetries of the system \cite{goldenfeld1992lectures,nishimori2010elements}. Furthermore, universal relations that transcend specific symmetry classes, exemplified by the Rushbrooke inequality, have also been firmly established \cite{nishimori2010elements}. Given the established framework of equilibrium phase transitions, it is a natural progression to investigate whether analogous principles govern entropy production in nonequilibrium phase transitions. In this paper, we address this question by focusing on homogeneous reversible chemical reactions as a prototypical nonequilibrium phenomenon. Just as universality classes have been characterized in equilibrium spin systems, we derive the universality classes and universal relations for entropy production within the chemical reactions. 

We focus on the two quantities related to entropy production, i.e., variance of entropy production ${\rm Var}\sigma$ and the parametric response of the entropy production rate $\partial_{\theta}\sigma$ when the concentration of a chemostatted chemical species is varied as a control parameter $\theta$:
\begin{align}
\begin{split}
{\rm Var}\sigma &\propto |\theta - \theta_c|^{-\alpha_{\pm}} \, , \\
\partial_{\theta}\sigma &\propto |\theta - \theta_c|^{-\beta_{\pm}} \, , 
\end{split}
\label{alphabetadef}
\end{align}
where the precise definition of ${\rm Var}\sigma$ and $\partial_{\theta}\sigma $ is defined below in Eqs.(\ref{defofvarsigma}) and (\ref{defofresponse}). The value $\theta_c$ is the critical point, and $\alpha_{\pm}$ and $\beta_{\pm}$ are the critical exponents. The subscript $\pm$ implies the exponents for $\theta > \theta_c$ and $\theta < \theta_c$, respectively. Employing the mathematically well-established {\it center manifold theory} \cite{Kuznetsov2023Elements}, we complete generic list of the critical exponents across archetypal bifurcations classified at the lowest order of the normal form, i.e., pitchfork, transcritical, saddle-node, and Hopf bifurcations (Table \ref{tab:critical-exponent}). Here, ``{\it generic}'' implies the general case, excluding exceptional scenarios such as the vanishing of coefficients for diverging terms. Diverging behavior in the fluctuation of entropy production as well as in the parametric response are governed by the stability properties of the stable solutions of the deterministic equations in the macroscopic-scale (thermodynamic) limit. Especially for the Hopf bifurcation, we employ the Floquet theory \cite{chicone1999ordinary} to find the divergent behavior in the fluctuation. The table \ref{tab:critical-exponent} shows that the exponents can be asymmetric across the critical point. 

We also derive a universal inequality valid for any types of bifurcation beyond lowest-order normal forms in the central manifold theory, i.e., $\alpha - 2 \beta \ge 0$ (Eq. (\ref{cramerraoineq})). The inequality clarifies the generic rule that if the parametric response diverges, the fluctuation of entropy production must strictly diverge. The converse, however, is not necessarily true. This implies that fluctuations reflect criticality more sensitively than the parametric response. Our findings extend the concept of universality in critical phenomena to the entropy production.

\sectionprl{Setup of homogeneous reversible chemical reaction}
We consider a homogeneous chemical reaction network, where the state of the system is fully characterized by the discrete numbers of the chemical species:
\begin{align}
 \sum_{\ell' \in \mathcal{S}_c} \nabla_{+\rho}^{\ell '} A_{\ell'} +  \sum_{\ell \in \mathcal{S}} \nabla_{+\rho}^{\ell} X_{\ell}  \overset{k_{+\rho}}{\underset{k_{-\rho}}{\rightleftharpoons}} \sum_{\ell \in \mathcal{S}} \nabla_{-\rho}^{\ell} X_{\ell} &+ \sum_{\ell' \in \mathcal{S}_c} \nabla_{-\rho}^{\ell'} A_{\ell'} \, , 
\end{align}
where $X_{\ell}$ denotes the $\ell$th dynamical chemical species whose concentration changes in time, while $A_{\ell'}$ represents the $\ell'$th chemostatted species whose concentration is fixed. The sets of these species are denoted by $\mathcal{S}$ and $\mathcal{S}_c$, respectively. Let $\mathcal{R}$ be the set of all reactions, indexed by $\rho \in \mathcal{R}$. The quantities $\nabla_{\pm\rho}^\ell \in \mathbb{Z}_{\ge 0}$ and $k_{\pm\rho}$ denote the stoichiometric coefficients and the rate constants for the forward ($+$) and backward ($-$) reaction in the reaction channel $\rho$ which are set to positive values ($k_{\pm\rho}>0$). Throughout this paper, we set the Boltzmann constant and temperature to unity ($k_{\text{B}} = T = 1$).

Let $\Omega$ be the volume of the reactor. In the microscopic picture of a finite $\Omega$, the dynamics are well described by a stochastic process for the discrete numbers of chemical species. Let $n_{\ell}$ ($\ell \in \mathcal{S}$) be the number of molecules of the $\ell$th chemical species. The probability distribution $P(\bm{n}, t)$ at time $t$ obeys the chemical master equation \cite{gardiner2009stochastic}:
\begin{align}
\frac{\partial}{\partial t}P(\bm{n}, t) &=
\sum_{\rho \in \mathcal{R}} \sum_{s=\pm} \Bigl[ W^{\rho}_{\bm{n}, \bm{n}+s\nabla_{\rho}}  P(\bm{n}+s \nabla_{\rho},t) \nonumber \\
&\quad - W^{\rho}_{\bm{n}+s\nabla_{\rho},\bm{n}}  P(\bm{n},t) \Bigr] \, , 
\label{microprobeq}
\end{align}
where $\nabla_{\rho}$ is the state-change vector, the $\ell$th element of which is $\nabla_{\rho}^{\ell}:=\nabla_{-\rho}^{\ell}-\nabla_{+\rho}^{\ell}$. The transition rates are 
\begin{align}
W^{\rho}_{\bm{n}, \bm{n} \mp \nabla_{\rho}} &= \Omega k_{\pm\rho} \prod_{\ell' \in \mathcal{S}_c} a_{\ell'}^{ \nabla_{\pm\rho}^{\ell'}}\prod_{\ell \in \mathcal{S}} \frac{(n_\ell \mp \nabla_{\rho}^{\ell} )!}{ (n_\ell - \nabla_{\mp\rho}^{\ell})!} \Omega^{- \nabla_{\pm\rho}^{\ell}} \, , \nonumber
\label{wrhonn}
\end{align}
where $a_{\ell'}$ is the fixed concentration of the chemostatted species $A_{\ell'}$ ($\ell' \in \mathcal{S}_c$).

Consider a trajectory $\Gamma$ from time $t=0$ to $t=\tau$, $\Gamma : \bm{n}_0 \xrightarrow{(\rho_1, s_1, t_1)} \bm{n}_1  \cdots \xrightarrow{(\rho_i, s_i, t_i)} \bm{n}_i  \cdots \xrightarrow{(\rho_{r}, s_{r}, t_{r})} \bm{n}_{r} \, , $ where $t_i$ denotes the time at which reaction $\rho_i$ occurs with $s_i = \pm 1$ (forward or backward reaction), and $r$ is the total number of reactions in the trajectory. The vector $\bm{n}_i$ represents the state just after the reaction $\rho_i$, and $\bm{n}_0$ is the initial number vector. We define the trajectory-dependent environment entropy production as \cite{seifert2025stochastic, Horowitz10.1063/1.4927395DiffusionApproximation}:
\begin{align}
\Sigma_{\rm e} (\Gamma ) &:=\sum_{\rho\in {\cal R}} \mu_{\rho} Z_{\rho} (\Gamma ) \, , ~~~~Z_{\rho} (\Gamma ) :=\sum_{i=1}^{r} s_i\,\delta_{\rho, \rho_i}\, , 
\end{align}
where $\delta_{\rho, \rho_i}$ is the Kronecker's delta, and $Z_{\rho} (\Gamma )$ is the extent of the reaction $\rho$. The quantity $\mu_{\rho}$ is the entropy increment via reaction channel $\rho$, which is defined as $\mu_{\rho} = \ln \bigl[ (k_{+\rho}/k_{-\rho})\prod_{\ell' \in \mathcal{S}_c} a_{\ell'}^{(\nabla_{+\rho}^{\ell'} - \nabla_{-\rho}^{\ell'})} \bigr]$, and hence $\Sigma_{\rm e}$ represents the entropy flowing into the thermal environment through the exchange of chemostatted species.

In the macroscopic-scale limit $\Omega\to\infty$, the concentration of the chemical species is defined as $x_{\ell} (t) = \lim_{\Omega\to\infty} \Omega^{-1} \sum_{\bm{n}'} n_{\ell} P(\bm{n}', t)$, which obeys the deterministic equation based on the law of mass action \cite{GrootBA03510754NonequilibriumThermo, Epstein10.1093/oso/9780195096705.001.0001NonlinearChemicalDynamics}
\begin{align}
\begin{split}
\dot{x}_{\ell} = F_{\ell} ({\bm x}(t)),& ~~~ F_{\ell} ({\bm x}(t)) = \sum_{\rho \in {\cal R}} \nabla_{\rho}^{\ell} (J_{+\rho} - J_{-\rho}) \, , \\
J_{\pm \rho} ({\bm x}(t))&:=k_{\pm \rho} \prod_{\ell' \in \mathcal{S}_c} a_{\ell'}^{ \nabla_{\pm\rho}^{\ell'}} \prod_{\ell \in \mathcal{S}} [x_{\ell}(t)]^{\nabla_{\pm\rho}^{\ell}} \, . 
\end{split}
\label{determeq}
\end{align}
The deterministic equation is rigorously derived from the microscopic stochastic description, as proven in \cite{Kurtz10.1063/1.1678692, Kurtz_1971}. The environment entropy production rate per unit volume is defined as
\begin{align}
\begin{split}
\sigma (\tau ) &= {d \over d\tau} \lim_{\Omega\to \infty} {1\over \Omega}\sum_{\rho\in{\cal R}} \mu_{\rho} \langle {Z}_{\rho} (\Gamma ) \rangle \, ,  \\
&=\!\sum_{\rho\in{\cal R}} \mu_{\rho} [J_{+\rho} ({\bm x}(\tau) ) -  J_{-\rho} ({\bm x}(\tau) ) ] =: \sigma ({\bm x}(\tau) ) \, , 
\end{split}
\end{align}
where $\langle ... \rangle$ denotes the ensemble average over trajectories. It is well-known that the macroscopic total entropy production rate per unit volume is given by $\sigma_{\rm m} (\tau) =\sum_{\rho} (J_{+\rho} ({\bm x}(\tau)) - J_{-\rho}({\bm x}(\tau)) \ln ({J_{+\rho}({\bm x}(\tau)) / J_{-\rho}({\bm x}(\tau) }))$ \cite{RaoPhysRevX.6.041064ChemicalThermo}, which is connected to environment entropy production rate as $\sigma_m (\tau)= \sigma(\tau) - \sum_{\ell\in {\cal S}}\dot{x}_{\ell}(\tau ) \ln x_{\ell} (\tau )$. The proof is provided in the Supplemental Material (SM) \cite{Supplemental}. In this paper, we consider realistic chemical reactions where the concentrations of all species are finite, i.e., $| {\bm x}(t) | < \infty$. This condition ensures that the quantites $J_{\pm} ({\bm x}(t))$, $\sigma (t)$, and $\partial_{x_{\ell}} ({\bm x}(t))$ are all finite.

We control the concentrations of the chemostatted species and define a control parameter $\theta$ as
\begin{align}
a_{\ell} &= a_{\ell}^{(0)} + \theta \hat{a}_{\ell} \, , 
\end{align}
where $\hat{\bm{a}} = (\hat{a}_{\ell})_{\ell \in \mathcal{S}_c}$ is a unit vector specifying the direction in which the concentrations of the chemostatted species vary from a reference state $\bm{a}^{(0)}$. The parameter $\theta$ represents the amplitude along this direction. We then investigate the critical exponents $\alpha_{\pm}$ and $\beta_{\pm}$, as defined in Eq. (\ref{alphabetadef}). Note that the values of these critical exponents are independent of the specific functional choice between $a_{\ell}$ and $\theta$, provided that the dependence is linearized near the bifurcation point \footnote{For the choice of $a_{\ell} = a_{\ell}^{(0)} + g(\theta) \hat{a}_{\ell}$, the critical behavior shows $[g(\theta) -g(\theta_c)]^{-\alpha}\sim [g^{\prime}(\theta_c)]^{-\alpha} (\theta - \theta_c)^{-\alpha}$, which results in the same exponent $\alpha$.}

\sectionprl{Formula of fluctuation of entropy production}
We use the trajectory-dependent environment entropy to define the fluctuation of the entropy production in the macroscopic-scale limit:
\begin{align}
{\rm Var}\sigma &= \lim_{\tau\to\infty}{1\over \tau} \lim_{\Omega\to\infty} {1\over \Omega} {\rm Var} \Sigma_{e} \, , \label{defofvarsigma}
\end{align}
where $\mathrm{Var} \Sigma_{\text{e}}$ is the variance of the trajectory-dependent environment entropy evaluated over the ensemble of stochastic trajectories. Note that we take the macroscopic-scale limit $\Omega \to \infty$ prior to the long-time limit $\tau \to \infty$, which is the standard procedure for analyzing critical behavior near phase transitions \cite{nishimori2010elements}. Relying on the exact analysis by Kurtz \cite{Kurtz10.1063/1.1678692,Kurtz_1971}, which justifies the system size expansion \cite{gardiner2009stochastic}, we consider the asymptotically exact decomposition $n_{\ell} (\Gamma)/\Omega = x_{\ell} (t) + y_{\ell}(t)/\sqrt{\Omega} + \dots$ and $Z_{\rho} (\Gamma)/\Omega = z_{\rho} (t) + w_{\rho}(t)/\sqrt{\Omega} + \dots$, from which the Langevin equations for the variables $y_{\ell}$ and $w_{\rho}$ are derived \cite{Horowitz10.1063/1.4927395DiffusionApproximation}. Using these equations, as briefly outlined in the End Matter (EM) and detailed in the SM, we derive the following formulas (\ref{fpformula}) and (\ref{lcformula}) for the fluctuation of the entropy production in the macroscopic-scale limit. When the stable steady solution of Eq. (\ref{determeq}) is a fixed point (FP) vector, such as in pitchfork, transcritical, and saddle-node bifurcations \cite{Strogatz1994, Kuznetsov2023Elements}, the formula reads:
\begin{align}
\begin{split}
{\rm Var}\sigma &= \sum_{\rho\in {\cal R}} v_{\rho}^2(\bar{\bm x}) {\cal A}_{\rho}(\bar{\bm x})\, ~~~~~~~~~~~~~~~~~~~~~~~~~~[{\rm FP}]\, ,  \label{fpformula}\\ 
v_{\rho} &= \mu_{\rho} -\sum_{\ell, \ell' \in {\cal S}} {\partial_{{x}_{\ell}} \sigma (\bar{\bm x})} [{\mathbb S}^{-1} (\bar{\bm x} )]_{\ell, \ell'} \nabla_{\rho}^{\ell '} \, , 
\end{split}
\end{align}
where $\bar{\bm{x}}$ is the stable fixed point vector of Eq. (\ref{determeq}), and $\mathbb{S} (\bar{\bm{x}})$ is the stability matrix (Jacobian) evaluated at the fixed point, with elements $[\mathbb{S} (\bar{\bm{x}})]_{\ell, \ell'} = \partial F_{\ell} / \partial x_{\ell'} |_{\bm{x}=\bar{\bm{x}}}$. The kinetic weight factor $\mathcal{A}_{\rho}(\bar{\bm{x}})$ is defined as $\mathcal{A}_{\rho}(\bar{\bm{x}}) := J_{+\rho}(\bar{\bm{x}}) + J_{-\rho}(\bar{\bm{x}})$. Conversely, when the stable steady solution of Eq. (\ref{determeq}) is a limit cycle (LC) \cite{Kuznetsov2023Elements, Strogatz1994}, typically arising from a Hopf bifurcation, the formula is given by:
\begin{align}
\begin{split}
{\rm Var}\sigma &= \! \lim_{\tau\to\infty}\!{1\over \tau}\! \int_{0}^{\tau} \!\! ds \! \sum_{\rho}\! v_{\rho}^2(\bar{\bm x}(s)) {\cal A}_{\rho}(\bar{\bm x}(s))~[{\rm LC}] , \label{lcformula}\\ 
v_{\rho} (\bar{\bm x}(s)) &=\!\mu_{\rho} + \!\!\sum_{\ell, \ell' } \!\int_s^{\tau} \!\! dt \partial_{x_{\ell}} \sigma(\bar{\bm x}(t)) \Phi_{\ell, \ell'} (t,s) \nabla_{\rho}^{\ell'} , 
\end{split}
\end{align}
where $\Phi(t,s) = \mathcal{T}\exp\left[\int_{s}^{t} du \, \mathbb{S} (\bar{\bm{x}}(u))\right]$ is the fundamental matrix solution with the time-ordering operator $\mathcal{T}$. Eq. (\ref{fpformula}) can be directly recovered from Eq. (\ref{lcformula}) by assuming a time-independent solution. Crucially, Eq. (\ref{fpformula}) implies that the generic divergence properties are governed by the eigenvalues of the stability matrix. Using the {\it center manifold theory}, one can show that the spectrum of the stability matrix and the fundamental matrix are given by the center manifold and stable+unstable manifolds (See (S.59) and (S.70) in the SM \cite{Supplemental}). Notably, the stability matrix possesses a zero eigenvalue at the bifurcation point for FP cases, whereas for a Hopf bifurcation, the relevant eigenvalues remain non-zero (purely imaginary). These spectral properties are essential for determining the critical exponents $\alpha_{\pm}$, as explained below.

\sectionprl{Generic diverging behavior of fluctuation}
Based on Eqs. (\ref{fpformula}) and (\ref{lcformula}), we now discuss the \textit{generic} critical exponents $\alpha_{\pm}$. Let $\lambda$ be the eigenvalue of the stability matrix in (\ref{fpformula}), and let ${\bm r}$ and ${\bm l}$ be the corresponding right and left eigenvectors. Near the bifurcation point, we have $v_\rho = \sum_{\lambda,\ell,\ell'}  { r_{\ell}{\partial_{{x}_{\ell}} \sigma (\bar{\bm x})} \nabla_{\rho}^{\ell'} l_{\ell'} / \lambda } + O(1)$. Thus the leading divergence of ${\rm Var}\,\sigma$ is present when $C_\sigma :=(\sum_{\ell} r_{\ell}{\partial_{{x}_{\ell}} \sigma (\bar{\bm x})} )^2 \sum_{\rho}A_\rho(\bar {\bm x}) (\sum_{\ell'}\nabla_{\rho}^{\ell'} l_{\ell'} )^2$ is nonzero. Throughout this paper, ``{\it generic}'' means that such nondegeneracy coefficients do not vanish. If they vanish because of symmetry or parameter tuning, the leading exponent can be reduced. In our derivation, we employ mathematically established normal forms from \textit{center manifold theory} for each bifurcation class \cite{Kuznetsov2023Elements}. The results are summarized in Table \ref{tab:critical-exponent}. To illustrate the derivation, consider the supercritical pitchfork bifurcation. We use the {\it central manifold theory} where the stable solution is a fixed point in a one-dimensional effective phase space; thus, we apply Eq. (\ref{fpformula}). Eigenvalues in the stability matrix arising from the normal form in the center manifold (Table \ref{tab:critical-exponent}) is $\lambda_{+} = -2(\theta - \theta_c)$ for $\theta > \theta_c$, and $\lambda_- = -(\theta_c - \theta)$ for $\theta < \theta_c$. This directly results in the \textit{generic} exponents $(\alpha_+, \alpha_-) = (2,2)$ \footnote{We emphasize that the $c\to -c$ symmetry of the pitchfork normal form does not, by itself, imply $\sigma(c)=\sigma(-c)$. The entropy production is a physical observable expressed in the original concentrations and reservoir affinities, and it generally contains an odd component in the center-manifold coordinate, $\sigma(c)=\sigma_0+s_1c+s_2c^2+\cdots$ with $s_1\neq0$. The pitchfork exponents in Table I correspond to this generic case. Only when the full chemical reaction network has an additional exact symmetry that exchanges the two branches and leaves the environment entropy production invariant is the odd coefficient $s_1$ forbidden. Such symmetry-protected cases are nongeneric in the present sense and may exhibit reduced exponents.}. 
See the SM \cite{Supplemental} for mathematically rigorous computation. The similar computations are applied to the transcritical and saddle-node bifurcations to get the exponents in Table \ref{tab:critical-exponent}.

The most non-trivial analysis involves the Hopf bifurcation, where a limit cycle emerges in the regime $\theta > \theta_c$. Here, we utilize Eq. (\ref{lcformula}) in conjunction with the normal form on the two-dimensional center manifold. Because the limit cycle is a periodic solution, we apply Floquet theory. By extracting the relevant divergent contributions as outlined in the EM and detailed in the SM \cite{Supplemental}, we obtain the exponent $\alpha_+ = 1$. On the other hand, in the regime $\theta < \theta_c$, the stable solution is a fixed point, warranting the use of Eq. (\ref{fpformula}). In a Hopf bifurcation, the relevant eigenvalues of the stability matrix $\mathbb{S}$ never vanish; rather, as understood from the normal form in Table \ref{tab:critical-exponent}, the eigenvalues associated with the center manifold converge to purely imaginary values $\pm i\omega_0$ as $\theta \to \theta_c$. Consequently, fluctuations of the entropy production do not diverge for $\theta < \theta_c$, yielding $\alpha_- = 0_-$, where ``\,$0_-$" is a symbol indicating nonpositive values. Thus, the generic exponents exhibits an asymmetry, as depicted schematically in Table \ref{tab:critical-exponent}. In the EM, we present a numerical demonstration of these phenomena using the Brusselator model for the Hopf bifurcation \cite{Prigogine10.1063/1.1668896, kondepudi2014modern} as well as the model showing the transcritical bifurcation .

\sectionprl{Generic diverging behavior of parametric response of entropy production rate} 
We now discuss the {\it generic} critical exponents $\beta_{\pm}$ for the response of the macroscopic entropy production rate, defined in Eq. (\ref{alphabetadef}) as:
\begin{align}
{\partial_{\theta} \sigma}&:={\partial \over \partial \theta } \lim_{\tau \to \infty}{1\over \tau} \int_0^{\tau} 
dt \, \sigma (t)\, . \label{defofresponse}
\end{align}
With this definition, we evaluate the time average of the instantaneous entropy production rate, which is well-defined even for the limit cycle. For fixed points, taking the long-time average is unnecessary, as the entropy production rate is inherently constant in the steady state. 

In the macroscopic-scale limit, the time-averaged entropy production rate depends explicitly on the control parameter $\theta$, as well as implicitly on $\theta$ through the stable steady solution $\bar{\bm{x}}$. Denoting this steady-state rate as $\sigma(\theta, \bar{\bm{x}})$, we can apply the chain rule to expand the response: $\partial_{\theta} \sigma = \frac{\partial \sigma(\theta, \bar{\bm{x}})}{\partial \theta} + \frac{\partial \bar{\bm{x}}}{\partial \theta} \cdot \frac{\partial \sigma(\theta, \bar{\bm{x}})}{\partial \bar{\bm{x}}} \, .$ We can rigorously show that the divergent behavior arises solely from the factor $\partial \bar{\bm{x}} / \partial \theta$ in the second term, since all other terms remain strictly finite under the condition $|\bm{x}(t)| < \infty$. Using the normal form structures on the center manifold, we evaluate the \textit{generic} exponents $\beta_{\pm}$, as listed in Table \ref{tab:critical-exponent}. For instance, in the pitchfork bifurcation, the stable fixed point is $\bar{c} = \pm \sqrt{\theta - \theta_c}$ for $\theta > \theta_c$ and $0$ for $\theta < \theta_c$. Consequently, the {\it generic} response exponents must be $(\beta_+ , \beta_-) = (1/2, 0_-)$.

In contrast to the fixed-point cases, the Hopf bifurcation requires a careful treatment due to the emergence of a limit-cycle solution. For $\theta>\theta_c$, the dynamics on the center manifold is described by a periodic orbit with amplitude scaling as $\sqrt{\theta-\theta_c}$, and thus the entropy production rate must be evaluated through a time average over one period. Owing to the mass-action structure, the instantaneous entropy production rate is expressed as a polynomial of the state variables, which themselves are proportional to $\sqrt{\theta-\theta_c}$ multiplied by trigonometric functions. Let $\omega$ be an angular frequency $\omega=2\pi/T$ with the period of the limit cycle $T$. Upon time averaging, all terms containing odd powers of $\cos(\omega t)$ or $\sin(\omega t)$ vanish, and only even-order contributions survive. As a result, the averaged entropy production rate depends solely on integer powers of $\theta-\theta_c$, and no singular contribution arises from half-integer scaling. Consequently, parametric response never shows the divergence across the transition, and hence we have  $(\beta_-,\beta_+)=(0_-,0_-)$ \cite{Supplemental}.

\sectionprl{Robust inequality between the exponents $\alpha$ and $\beta$}
Up to this point, we have focused on the most fundamental bifurcations, whose normal forms are governed by the lowest-order non-linear terms (Table \ref{tab:critical-exponent}). In principle, however, one can realize an infinite number of bifurcations by fine-tuning the system parameters such that higher-order terms become the leading contributions in the normal form. Here, we establish a robust relationship between the critical exponents $\alpha$ and $\beta$ that remains valid for any type of bifurcation. Our analytical strategy relies on a Cram\'{e}r-Rao-type bound evaluated in the trajectory space, utilizing the trajectory-dependent environment entropy production. While the detailed calculations are deferred to the SM \cite{Supplemental}, we show the inequality
\begin{align}
|\partial_{\theta} \sigma - r(\theta)|^2 & \le i(\theta) \, {\rm Var}\sigma \, , 
\end{align}
where $r(\theta) = \lim_{\tau\to\infty} \tau^{-1} \lim_{\Omega\to\infty} \Omega^{-1} \langle \partial_{\theta} \Sigma_{\text{e}} (\Gamma) \rangle$, and $i(\theta)$ is the scaled Fisher information per unit volume and per unit time. One can readily show that $|r(\theta)| < \infty$ and $|i(\theta)| < \infty$ from the condition that the macroscopic concentrations remain finite, i.e., $|\bm{x}(t)| < \infty$ \cite{Supplemental}. Finiteness of these terms immediately leads to the fundamental inequality relating the exponents $\alpha$ and $\beta$ as
\begin{align}
\alpha - 2 \beta &\ge 0 \, , \label{cramerraoineq}
\end{align}
for $\alpha \ge 0$ and $\beta \ge 0$. One can readily verify that the generic exponents listed in Table \ref{tab:critical-exponent} satisfy this relation. Furthermore, for the Schl\"ogl model \cite{Schlogl1972, Vellela10.1098/rsif.2008.0476Schlogl}, which exhibits a cusp bifurcation, recent numerical computations along a path through the critical point estimated the critical exponents to be $\alpha = 1.3 \pm 0.2$ and $\beta = 0.65 \pm 0.05$ \cite{Remlein10.1063/5.0203659}. These results strongly suggest that the equality in relation (\ref{cramerraoineq}) is saturated. The inequality  means that if $\partial_{\theta}\sigma$ diverges (i.e., $\beta >0$), ${\rm Var}\sigma$ must strictly diverge (i.e., $\alpha >0$). The converse, however, is not necessarily true, as corroborated by the specific critical exponents in Table \ref{tab:critical-exponent}. This implies that fluctuations reflect criticality more sensitively than parametric responses.

\sectionprl{Concluding remarks}
In this paper, we elucidate the generic critical behavior of entropy production at macroscopic bifurcation points. By deriving general formulas for the entropy production variance Eqs. (\ref{fpformula})-(\ref{lcformula}) and the parametric response (\ref{defofresponse}), we derive the universality-class classification on the critical exponents $\alpha$ and $\beta$ using \textit{center manifold theory}, as shown in Table \ref{tab:critical-exponent}. The general mechanism for divergence was also clarified. Moreover, we establish a universal inequality valid for any bifurcation (Eq. (\ref{cramerraoineq})), dictating that a diverging response ($\beta > 0$) strictly necessitates diverging fluctuations ($\alpha > 0$), but not vice versa. This indicates that fluctuations are more sensitive indicators of criticality than responses.

We finally remark on the order of limits between the volume scale $\Omega$ and the observation time $\tau$. We here take the macroscopic-scale limit $\Omega\to\infty$ first before taking long observation time, as in (\ref{defofvarsigma}) and (\ref{defofresponse}). On the other hand, recent study for the Schl\"ogl model reports the exponentially large fluctuation of the entropy production with respect to the size for the saddle-node bifurcation \cite{NguyenPhysRevE.102.022101Schlogl}, which never occurs in our setup. This effect emerges because they look at the long time limit first before taking the macroscopic-scale limit. The critical behavior can be sensitive on the relative scales of observation time and system size, which somewhat reminds us of Keizer’s paradox in the chemical reaction \cite{keizer1979nonequilibrium,vellela2007quasistationary}. It is intriguing to elucidate the observation-time dependence in the critical phenomena.

\section*{Acknowledgments}
We are supported by JSPS KAKENHI Grant No. JP23K25796, No. JP26H02015, and JP26H00388.

\bibliography{cite}

\clearpage
\section*{End Matter}


\sectionprl{Variance formulas (\ref{fpformula}) and (\ref{lcformula})}
We outline the derivation (\ref{lcformula}) for the limit cycle case, which is reduced to (\ref{fpformula}) for the fixed point cases. The details are explained in the SM \cite{Supplemental}. For the size expansion $n_{\ell} (\Gamma ) /\Omega  \to x_{\ell}(t) + y_{\ell}(t)/\sqrt{\Omega}$ and $Z_{\rho} (\Gamma) / \Omega \to z_{\rho}(t) + w_{\rho}(t)/\sqrt{\Omega}$ \cite{Kurtz10.1063/1.1678692,Kurtz_1971,gardiner2009stochastic, Horowitz10.1063/1.4927395DiffusionApproximation}, we derive the Fokker-Planck equation for the variable $y_{\ell}$. One can then identify the Langevin equations for the variables $y_{\ell}$ and $w_{\rho}$ as \cite{Supplemental}
\begin{align}
\begin{split}
\dot{y}_{\ell} &= \sum_{\ell'}{\partial {F_{\ell}} ({\bm x}(t)) \over \partial x_{\ell'}} y_{\ell'} + \sum_{\rho}\nabla_{\rho}^{\ell} \sqrt{{\cal A}_{\rho } ({\bm x} (t) )} \, \xi_{\rho}(t) \, , \\
\dot{w}_{\rho} &= \sum_{\ell'}{\partial {J}_{\rho} ({\bm x}(t)) \over \partial {x}_{\ell'}} y_{\ell'} + \sqrt{{\cal A}_{\rho} ({\bm x} (t) )} \, {\xi}_{\rho}(t) \, , \label{endmatter2}
\end{split}
\end{align} 
where $\xi_{\rho} (t)$ is a Gaussian white noise satisfying $\langle\!\langle \xi_{\rho} (t) \rangle\!\rangle=0$ and $\langle\!\langle \xi_{\rho} (t)\xi_{\rho'} (t')  \rangle\!\rangle=\delta_{\rho,\rho'} \delta (t-t')$ for the noise-average $\langle\!\langle ... \rangle\!\rangle$. The noise amplitude is given by $\mathcal{A}_{\rho } (\bm{x}(t)) = \left[ {J}_{+\rho}(\bm{x}(t)) + {J}_{-\rho}(\bm{x}(t)) \right] $. Formal solution of $w_{\rho} (t)$ is
\begin{align}
 {w}_{\rho}(\tau)  & = \sum_{\rho'}\int_{0}^\tau d s  \, G_{\rho , \, \rho'}(\tau, s ) \, \xi_{\rho' } (s) \, ,  
\end{align}
where the Green's function matrix ${\bm G}$ that is given by 
\begin{align}
 G_{\rho,\,\rho'}(\tau, s) &= \Bigl[ \int_{s}^{\tau} d t \sum_{\ell,\ell'}{{J}_{\rho}(\bm{x}(t))
\over \partial {x}_{\ell} } \Phi_{\ell,\ell'}(t, s) \nabla^{\ell'}_{\rho '}+ \delta_{\rho,\rho'} \Bigr] \nonumber \\
&\times \sqrt{\mathcal{A}_{\rho'}(\bm{x}(s))} \, \label{endmatterGreenFunc} \, .
\end{align}
The matrix $\Phi(t, s)$ is $\Phi(t, s):= {\cal T}\exp\Bigl[ \int_{s}^{t} d s' {\partial {\bm F}(\bm{x}(s')) \over \partial {\bm x} } \Bigr]$ where ${\cal T}$ is the time-ordering operator. The fluctuation of the entropy production is formally formulated as
\begin{align}
{\rm Var}\sigma &=\lim_{\tau\to\infty}{1\over \tau} \sum_{\rho, \rho'} \mu_{\rho} \langle\!\langle w_{\rho} (\tau) w_{\rho'} (\tau )\rangle\!\rangle \mu_{\rho'} \, . \label{endmattervarsigma}
\end{align}
Plugging the solution ${w}_{\rho}(\tau)$ into (\ref{endmattervarsigma}) and taking average over the noise leads to the formula (\ref{lcformula}) of fluctuation of entropy production. 

\sectionprl{Exponents $\alpha_{\pm}$ for the Hopf bifurcation}
We outline the derivation of the exponents $\beta_{\pm}$. See the SM also for more details. We first explain the case of $\theta > \theta_c$ where the limit-cycle appears. We use the Floquet theory approach. 
Let us consider the general setup, where the vector ${\bm x}$ has $n$-components ($n\ge 2$). At the bifurcation point, the manifold is decomposed into $T_{\bar{\bm x}_c}{\mathbb R}^n = E^{\rm c}\oplus E^{\rm su}$ ($\bar{\bm x}_c$: fixed point at $\theta_c$), where $E^{\rm c}$ and $E^{\rm su}$ are respectively the center and stable+unstable manifolds. The center manifold is the two-dimension in the Hopf bifurcation case. We consider the monodromy matrix ${\bm M}$
\begin{align}
{\bm M} &= {\cal T} \exp \Bigl[\int_0^T dt \, {\mathbb S} ({\bm x} (t) ) \Bigr] \, ,
\end{align}
where the left and right eigenvectors for the $j$th Floquet exponent $\nu_j$ satisfy $\bm{w}^\dagger_j {\bm M} = e^{\nu_j T} \bm{w}^\dagger_j$ and ${\bm M} {\bm v}_i = e^{\nu_i T}{\bm v}_j$, respectively. Here, $T$ is the period of the limit cycle. Using the {\it center manifold theory}, one can rigorously show that $\nu_1=0$ and $\nu_2=O(\theta - \theta_c)$ from the center manifold $E^{\rm c}$ and the remaining $(n-2)$ exponents are $O(1)$ from the manifold $E^{\rm su}$ (See the lemma 1 in the SM \cite{Supplemental}). The matrix $\Phi(t,s)$ is decomposed into 
$\Phi(t, s)  =  \sum_{j=1}^n \Phi_{\theta,j}(t, s) $ and $\Phi_{\theta,j} (t,s) = e^{\nu_j (t - s)} \tilde{\bm{v}}_j(t) \tilde{\bm{w}}^\dagger_j(s) $, 
where $\tilde{\bm{w}}^\dagger_j(t) := e^{\nu_j t} \bm{w}^\dagger_j \Phi(0,t)$ and $\tilde{\bm{v}}_j(t) := e^{- \nu_j t} \Phi (t, 0) \bm{v}_j$. Through the computation with this decomposition, one can identify that the diverging term in the fluctuation is given from the component $j=2$ \cite{Supplemental}, which leads to the dominant contribution $[{\rm Var}\sigma]_{\rm D}^{(1)}$ written as
\begin{align}
[{\rm Var}\sigma]_{\rm D}^{(1)} &= { |{\cal F}_{2,0} (\theta ) |^2 \over |\nu_2|^2} {\cal D}_{0}(\theta ) \, , \\
{\cal F}_{2,0} (\theta ) &= \int_0^{T} dt  \sum_{\ell}{\partial \sigma  (\bar{\bm x}(t)) \over \partial x_{\ell}} [\tilde{\bm v}_2 (t)]_{\ell} \, , \nonumber \\
{\cal D}_0 (\theta ) &= \int_0^T d t\, \sum_{\rho ,\ell,\,\ell'}[\tilde{\bm w}_2^{\dagger} (t)]_{\ell} \nabla_{\rho}^{\ell} {\cal A}_{\rho} (\bar{\bm x} (t)) \nabla_{\rho}^{\ell'} [\tilde{\bm w}_2^{\dagger} (t)]_{\ell'}^{\ast} \, . \nonumber
\end{align}
Utilizing the lemma $2$ and $3$ in the SM \cite{Supplemental}, one finds the numerator is $O(\theta - \theta_c)$, and hence this term becomes $O(1/(\theta - \theta_c))$ leading to $\alpha_+=1$.  

In the case of $\theta < \theta_c$ where the stable solution is a fixed point vector, we use the formula (\ref{fpformula}). Using the center manifold theory, one can show that the inverse of the stability matrix is given by the contributions from $E^{\rm c}$ and $E^{\rm su}$ (See Eq. (S59) in the SM \cite{Supplemental}). The stability matrix arising from the center manifold  is given as
\begin{align}
\tilde{\mathbb S}^{({\rm c})} = 
\left( 
\begin{array}{cc}
\theta - \theta_c , & -\omega_0 \\
\omega_0 , & \theta - \theta_c
\end{array}
\right) \, . 
\end{align}
Hence, no singularity exists in the inverse of the stability matrix. This means that no diverging behavior exists in this regime, leading to $\alpha_-=0_-$, i.e., a nonpositive value.
\begin{figure}[ht]
\begin{center}
\includegraphics[scale=0.45]{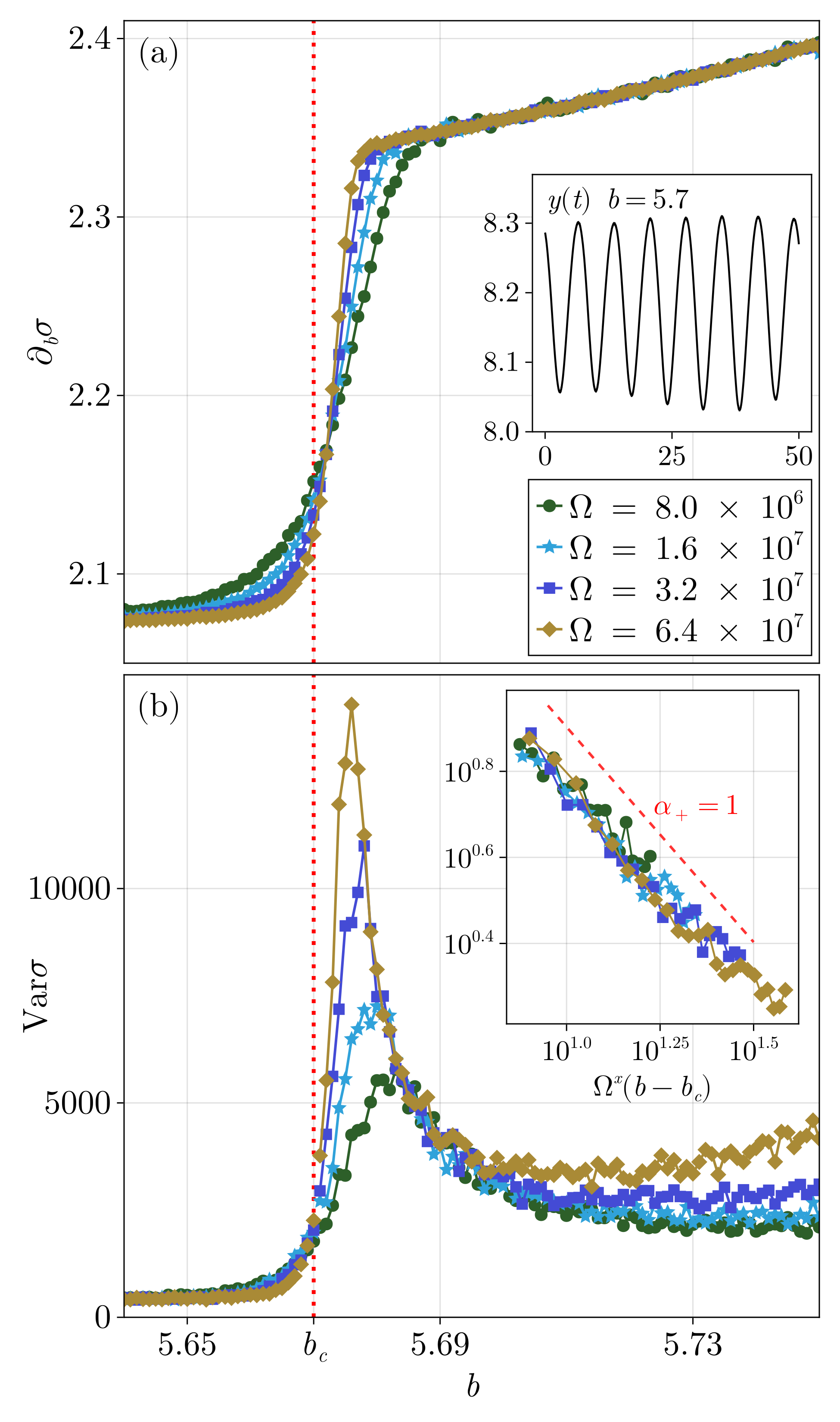}
\end{center}
\vspace{-0.5cm}
\caption{
(a): Parametric response $\partial_b \sigma$. The inset shows the time evolution of the concentration $y(t)$ of the chemical species $\mathrm{Y}$ at $b = 5.69$ in the oscillatory phase. The red vertical dotted line indicates $b_c$. A finite gradient at $b_c$ is observed, and hence $\beta_{\pm}=0_-$ (nonpositive values) are confirmed. (b): The variance $\mathrm{Var}\,\sigma$. The inset shows the power-law slope with exponent $\alpha_+ = 1$. To estimate the critical value $b_c$ indicated by the red dotted line, we use the ansatz $\mathrm{Var}\,\sigma = \Omega^{y} f_\sigma\!\left(\Omega^{x}(b - b_c)\right)$, with $b_c = 5.67$, $x = 0.40$, and $y = 0.42$. }
\label{fig:left}
\end{figure}

\sectionprl{Hopf bifurcation of the Brusselator model}
To verify the validity of our theoretical framework, we present numerical results for the Brusselator, a quintessential model exhibiting a supercritical Hopf bifurcation. The reaction scheme of the Brusselator is described as follows:
\begin{align}
\begin{split}
& A \overset{k_{+1}}{\underset{k_{-1}}{\rightleftharpoons}}   X, ~~~
B + X \overset{k_{+2}}{\underset{k_{-2}}{\rightleftharpoons}} 
 Y + D,  \\
& 2X + Y \overset{k_{+3}}{\underset{k_{-3}}{\rightleftharpoons}} 3X,  ~~
X \overset{k_{+4}}{\underset{k_{-4}}{\rightleftharpoons}} E
\label{eq:Brusselator}
\end{split}
\end{align}
where $X$ and $Y$ denote the intermediate chemical species, while $A, B, D,$ and $E$ represent the species maintained under chemostatic conditions \cite{Prigogine10.1063/1.1668896,kondepudi2014modern}. In our simulations, the rate constants are set to $k_{\pm i} = 1$ for $i= 1, \dots, 4$, and the concentrations of $A, D,$ and $E$ are fixed at $a = 1.0$ and $d = e = 0.1$, respectively. We employ the concentration of species $B$, denoted as $\theta = b$, as the bifurcation parameter. Within this parameter regime, the corresponding deterministic rate equations undergo a supercritical Hopf bifurcation at the critical threshold $b_c = 5.67537$. The chemical Langevin equations associated with \eqref{eq:Brusselator} were numerically integrated using the Euler-Maruyama scheme \cite{higham2001algorithmic}. Fig. \ref{fig:left} illustrates the fluctuations in entropy production and the results of finite-size scaling performed within the oscillatory phase. Fig. \ref{fig:left}(a) indicates a finite gradient at $b=b_c$ meaning that no divergence occurs, and hence $\beta_{\pm}=0_-$ (nonpositive values). Fig.\ref{fig:left}(b) clearly shows the exponent $\alpha_{+}=1$ and $\alpha_-= 0_-$.

\sectionprl{Transcritical bifurcation of the chemical reaction model}
To investigate the system's behavior in the vicinity of a transcritical bifurcation, we consider a model governed by the chemical reaction as
\begin{align}
&\text{A} \overset{1}{\underset{1}{\rightleftharpoons}} \text{X + B}\, ,~~\text{D + X}
 \overset{1}{\underset{1}{\rightleftharpoons}} \text{2X}  \, , ~~
 \text{E + 2X}  \overset{1}{\underset{1}{\rightleftharpoons}} \text{3X}. \label{eq:transcriticalModel}
\end{align}
We set the concentrations for chemostatted species satisfying the relation $b = 5a/2  + 2$, $d = 1$, and $e = a + 7/2$, which leads to the deterministic equation
$\dot{x} =- (x - 1/2 )(x - a)(x - 2)$. The transformations $a=2+(2/3)\theta$ and $x= x +(2/3)c$ lead to the effective equation $\dot{c} = \theta c - c^2$ for the regime $0\le c \ll 1$, which is identical to the normal form of the transcritical bifurcation ($\theta_c=0$) in Table \ref{tab:critical-exponent}. 

We change the parameter $a$ near the critical value $a_c=2$. Note that for $a>a_c$, the stable fixed point $\bar{x}$ is $\bar{x}=a$, while for $1/2<a<a_c$, $\bar{x}=2$. We present the numerical simulation results in Fig.\ref{fig:trans}, which clearly shows the exponent $\alpha_{+}=2$ in agreement with the Table \ref{tab:critical-exponent}. We also check $\alpha_-=2$. The behavior of the fluctuation is almost symmetric with respect to $a_c$. In the main plot of Fig.\ref{fig:trans}, a sharp peak is observed near the transition point. This peak is thought to arise because, near the transition point, the attractors $x=a$ or $x=2$ become nearly semistable, leading to transitions between these attractors and the stable attractor $x=1/2$. As the system size increases, this peak converges to a point just above the transition point. Therefore, this behavior is interpreted as a finite-size effect.
\begin{figure}[hb]
\begin{center}
\includegraphics[scale=0.45]{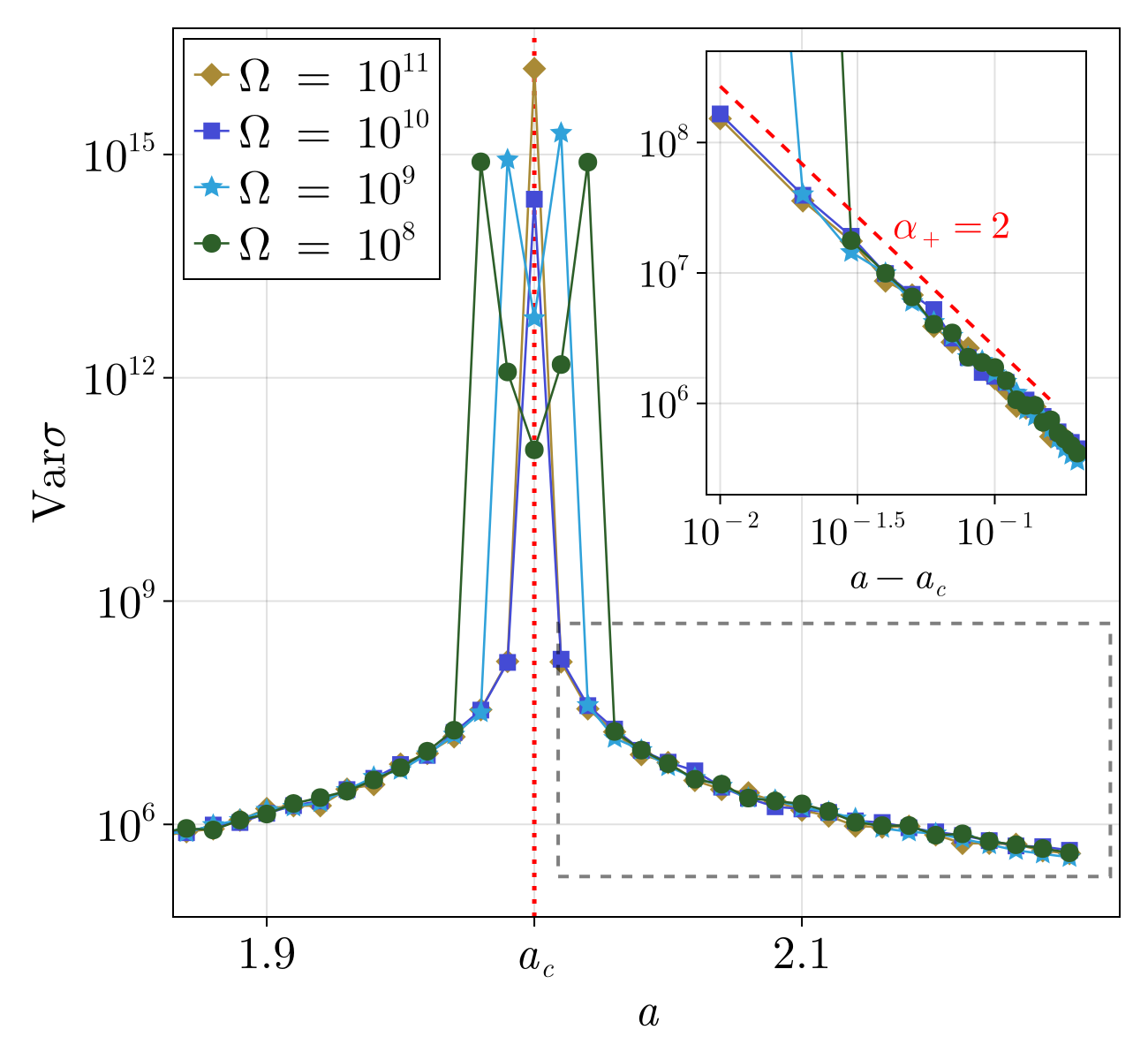}
\end{center}
\vspace{-0.5cm}
\caption{
Variance of the entropy production in the model \eqref{eq:transcriticalModel} as a function of $a$. The red vertical dotted line marks $a_c$. The inset shows a log-log plot of the region $a>a_c$ indicated by the dashed rectangle; only this side is shown because the behavior is nearly symmetric about $a_c$. Since the system size is sufficiently large, no finite-size scaling is performed. The red dashed line in the inset indicates the power-law slope $\alpha_+=2$.
}
\label{fig:trans}
\end{figure}

\clearpage

\setcounter{secnumdepth}{3} 

\setcounter{section}{0}
\setcounter{equation}{0}
\setcounter{footnote}{0}

\renewcommand{\thesection}{S.\Roman{section}}
\renewcommand{\theequation}{S.\arabic{equation}}
\renewcommand{\thefootnote}{*\arabic{footnote}}

\begin{widetext}

\begin{center}
{\large \bf Supplementary Material for \protect \\ ``Universal criticality of entropy production in chemical reaction networks''}\\
\vspace*{0.3cm}
Kyota Tamano$^{1}$ and Keiji Saito$^{1}$ \\
\vspace*{0.1cm}
 $^{1}${\small {\it Department of Physics, Kyoto University, Kyoto 606-8502, Japan}}
\end{center}


\section{Well-mixed chemical reaction}
\label{chemreac}
We consider a well-mixed reversible chemical reaction network, where spatial inhomogeneities are neglected so that the state of the system is fully characterized by the concentrations of the chemical species, without reference to spatial coordinates. The set of chemical species is decomposed into two subsets: the dynamical species ${\cal S}$, whose concentrations evolve in time and serve as the system variables, and the chemostatted species ${\cal S}_c$, whose concentrations are maintained constant through exchange with external reservoirs. We denote the species in ${\mathbb S}$ by $X_\ell$ $(\ell \in {\cal S})$, and those in ${\cal S}_c$ by $A_\ell$ $(\ell \in {\cal S}_c)$. Let ${\cal R}$ denote the set of reactions. Each reaction $\rho \in {\cal R}$ is represented in the form
\begin{align}
 \sum_{\ell \in \mathcal{S}_c} \nabla_{+\rho}^{\ell} A_\ell +  \sum_{\ell \in \mathcal{S}} \nabla_{+\rho}^{\ell} X_{\ell}
 \overset{k_{+\rho}}{\underset{k_{-\rho}}{\rightleftharpoons}}  
    \sum_{\ell \in \mathcal{S}} \nabla_{-\rho}^{\ell} X_{\ell} + \sum_{\ell \in \mathcal{S}_c} \nabla_{-\rho}^{\ell} A_\ell \, , 
\end{align}
where $\nabla_{\pm\rho}^\ell \in \mathbb{Z}_{\ge 0}$ denotes the stoichiometric coefficients. We define the stoichiometric matrix, the matrix element of which is given as
\begin{align}
\nabla_{\rho}^\ell:= \nabla_{-\rho}^\ell - \nabla_{+\rho}^\ell \, . 
\end{align}
The matrix $\nabla$ is $|{\cal S}|\times |{\cal R}|$ matrix, where $|{\cal S}|$ and $|{\cal R}|$ are the cardinality of the sets ${\cal S}$ and ${\cal R}$, respectively. The quantity $\nabla_{\rho}^\ell$ is the net number of molecules of the chemical species $\ell$ produced when the reaction $\rho$ proceeds once in the forward direction. The quantity $k_{\pm\rho}$ denote the forward and backward reaction rate constants, respectively. Throughout this paper, we set $k_{\rm B}\equiv 1$ and $T=1$. 
\subsection{Microscopic probabilistic processes}
We now introduce the microscopic time-evolution equation of the chemical reaction system, i,e., the probabilistic process. Let ${\bm n} = (n_\ell)_{\ell \in {\cal S}}$ denote the vector of molecule numbers of the chemical species. The probability distribution $P({\bm n}, t)$ obeys the following chemical master equation.
\begin{align}
{\partial \over \partial t}P({\bm n}, t) &=
\sum_{\rho \in {\cal R}} \sum_{s=\pm} W^{\rho}_{{\bm n}, {\bm n}+s\nabla_{\rho}}  P({\bm n}+s \nabla_{\rho},t) - W^{\rho}_{{\bm n}+s\nabla_{\rho},{\bm n}}  P({\bm n},t)
 \, , \label{sm:microprobeq}
\end{align}
where $\nabla_{\rho}$ is a vector, the $\ell$th element of which is $\nabla_{\rho}^{\ell}$. Let $\Omega$ denote the system volume, and let $a_\ell$ be the concentration of the chemostatted species $A_\ell$. Then, the transition rates are given by the following expression:
\begin{align}
W^{\rho}_{\bm{n}, \bm{n} \mp \nabla_{\rho}} &= \Omega k_{\pm\rho} \prod_{\ell' \in \mathcal{S}_c} a_{\ell'}^{ \nabla_{\pm\rho}^{\ell'}}\prod_{\ell \in \mathcal{S}} \frac{(n_\ell \mp \nabla_{\rho}^{\ell} )!}{ (n_\ell - \nabla_{\mp\rho}^{\ell})!} \Omega^{- \nabla_{\pm\rho}^{\ell}} \, . \label{sm:wrhonn}
\end{align}

\subsection{Deterministic equations in the macroscopic scale}
In the macroscopic limit (thermodynamic limit) $\Omega\to\infty$, one defines the macroscopic value of chemical species 
\begin{align}
{\bm x}(t) &:=\lim_{\Omega\to\infty} {1\over \Omega }  \sum_{\bm n}{\bm n} P({\bm n},t) \, . 
\end{align}
This macroscopic concentration per unit volume obeys the deterministic equation based on the law of the mass action:
\begin{align}
\dot{\bm{x}}(t) & = {\bm F}({\bm x}(t)) =\sum_{\rho} \nabla_{\rho} J_{\rho} \, , \label{macro_deterministic_eqs}
\end{align}
where ${\bm{J}}$ and ${J}_{\rho}$ are defined as
\begin{align}
{\bm{J}} &:= ({J}_{\rho})_{\rho \in \mathcal{R}} \, , \\
{J}_{\rho} &:= {J}_{+\rho} - {J}_{-\rho} \, , \\
{J}_{\pm\rho}(\bm{x}(t)) 
    &:= k_{\pm \rho} \prod_{\ell' \in \mathcal{S}_c} a_{\ell'}^{ \nabla_{\pm\rho}^{\ell'}} \prod_{\ell \in \mathcal{S}} x_{\ell}^{\nabla_{\pm\rho}^{\ell}}(t) \, . 
\end{align}
We should note that Kurtz has proved the micro-macro correspondence by deriving the deterministic time-evolution from the microscopic probabilistic processes (\ref{sm:microprobeq} in the macroscopic limit \cite{Kurtz10.1063/1.1678692,Kurtz_1971}.
The terms $J_{+\rho}$ and $J_{-\rho}$ are respectively interpreted as the frequency of the forward and backward action in the reaction $\rho$ per unit volume. In addition, Kurtz's analysis identifies the relations 
\begin{align}
J_{\pm \rho} &=\lim_{\Omega\to\infty} {1\over \Omega} \sum_{{\bm n}} W^{\rho}_{{\bm n}\pm \nabla_{\rho},{\bm n}}  P({\bm n},t) \, . 
\end{align}
In this paper, we consider the chemical reactions having finite steady state vector $|{\bm x(t)}|<\infty$. 

\section{Entropy production: micro and macro expressions}\label{equivalence_entropy}
\subsection{Microscopic entropy production}
Let $\Gamma$ denote a trajectory from the initial time $0$ to the final time $\tau$:
\begin{align}
\Gamma &= {\bm n}_0 \xrightarrow{ (\rho_1, s_1, t_1)} {\bm n}_1 \xrightarrow{ (\rho_2, s_2, t_2)} {\bm n}_2 \xrightarrow{}  \cdots  \xrightarrow{}
{\bm n}_{i-1} \xrightarrow{ (\rho_i, s_i, t_i)} {\bm n}_{i} \cdots 
 {\bm n}_{r -1}  \xrightarrow{ (\rho_{r}, s_{r}, t_{r})} {\bm n}_{r}
\end{align}
where $t_i$ denotes the time at which the reaction $\rho_{i}$ occurs with the sign $s_i$. The sign $s_i=\pm 1$ implies the forward (backward) direction in the reaction $\rho_i$. The vector ${\bm n}_i$ is a number vector after the $i$th reaction $\rho_i$. We set ${\bm n}_0$ to a molecule number vector at the initial time $0$. Here $r$ is a total number of the reactions in this trajectory. For a given trajectory, we introduce the 
extend of the reaction $\rho$, denoted by $Z_\rho (\Gamma )$, which is defined explicitly as
\begin{align}
Z_{\rho}(\Gamma) &= \sum_{i=1}^{r} s_i\,\delta_{\rho,\,\rho_i}\, ,
\end{align}
where $\delta_{\rho,\,\rho_i}$ is the Kronecker's delta function. From the physical viewpoint, the trajectory-dependent environment entropy is defined as follows
\begin{align}
\Sigma_{\rm e} (\Gamma ) &=\sum_{\rho} \mu_{\rho} \, Z_{\rho} (\Gamma ) \, , \label{entropybath}
\end{align}
where $\mu_{\rho } $ is the entropy increment via reaction channel $\rho$:
\begin{align} 
\mu_{\rho} &= \ln \Bigl[ {k_{+\rho} \over k_{-\rho}}  \prod_{\ell' \in S^{\mathrm{c}}} 
a_{\ell'}^{- \nabla_{\rho}^{\ell'} } \Bigr]. \label{sp_murho}
\end{align}
Note that the total entropy production rate at the steady state is equivalent to the bath entropy production rate, since the system entropy production rate vanishes. In this paper, we use this physical expression (\ref{entropybath}) for computing the entropy production. 

On the other hand, the stochastic thermodynamics suggests that the environment contribution to entropy production is computed from the probability ratio between the forward and backward trajectories:
\begin{align}
\Sigma_{\rm stc} (\Gamma ) &:=   \sum_{i=1}^{N_r} \ln {W^{\rho_i}_{{\bm n}_{i}, {\bm n}_{i - 1}} \over  W^{\rho_i}_{{\bm n}_{i - 1}, {\bm n}_{i}} } \, . \label{app:eq:stochastic-entropy-production}
\end{align}
In this paper, however, we do not use this formulation.

\subsection{Macroscopic entropy production}
For the macroscopic chemical reaction where the dynamics obeys the deterministic motion (\ref{macro_deterministic_eqs}), the entropy production rate per unit volume is known to be 
\begin{align}
\sigma_m (\tau ) &= \sum_{\rho} \left[ {J}_{+\rho} ({\bm x}(\tau )) - {J}_{-\rho} ({\bm x}(\tau )) \right] \ln \frac{{J}_{+\rho}({\bm x}( \tau ))}{{J}_{-\rho} ({\bm x}(\tau ))} \, . \label{sp:def_macro_entropy}
\end{align}

\subsection{Relations between micro and macro entropy productions}
We consider relation between the entropy production rates using the deterministic equation and environment entropy. We define the entropy production rate for the environment entropy production per unit volume:
\begin{align}
\sigma (\tau ) &= {d\over d\tau} \lim_{\Omega \to \infty} {1\over \Omega} \langle \Sigma_{\rm e}(\Gamma) \rangle \, , 
\end{align}
where $\langle ... \rangle$ implies taking an average over path (trajectory) probability. 
One can rigorously show the following relation 
\begin{align}
\sigma (\tau)&= \sigma_m (\tau) + \sum_{\ell \in {\cal S}} \dot{x}_{\ell}(\tau) \ln x_{\ell} (\tau) \, .\label{micromacroentropyprod2}
\end{align}
The second term $\sum_{\rho}\sum_{\ell \in {\cal S}} \dot{x}_{\ell} (\tau ) \ln  x_{\ell}(\tau )$ is regarded as the contribution from the system's entropy. 

The relation (\ref{micromacroentropyprod2}) is readily proved using the micro-macro correspondence proven by Kurtz \cite{Kurtz10.1063/1.1678692,Kurtz_1971}. The environment entropy production per unit volume is given as $\lim_{\Omega \to \infty} \frac{1}{\Omega} \langle \Sigma_{\rm e} (\Gamma) \rangle = \sum_{\rho} \mu_{\rho} \lim_{\Omega \to \infty} \frac{1}{\Omega} \langle Z_{\rho} \rangle $, where 
\begin{align}
\langle Z_{\rho} (\Gamma ) \rangle &= \Bigl< \sum_i s_i\, \delta_{\rho, \rho_i}\Bigr>=\int_{0}^{\tau} dt \sum_{\bm{n}'} \left( W^{\rho}_{\bm{n}'+ \nabla_{\rho}, \bm{n}'} -W^{\rho}_{\bm{n}' - \nabla_{\rho}, \bm{n}'} \right)
P({\bm n}', t).
\end{align}
By dividing this by $\Omega$ and taking the macroscopic limit $\Omega\to\infty$, we have 
\begin{align}
\lim_{\Omega \to \infty} \frac{1}{\Omega} \langle Z_{\rho} (\Gamma )\rangle &= \int_{0}^{\tau} dt \,  {J}_{\rho}(\bm{x}(t)) \, , 
\end{align}
where we have used the micro-macro correspondence which is mathematically proven by Kurtz \cite{Kurtz10.1063/1.1678692,Kurtz_1971}. This leads to 
\begin{align}
\sigma (\tau) &= {d\over d \tau} \lim_{\Omega \to \infty} \frac{1}{\Omega} \langle \Sigma_{\rm e} (\Gamma) \rangle = \sum_{\rho} \mu_{\rho} {J}_{\rho}(\bm{x}(\tau )) \nonumber \\
&= \sum_{\rho} \Bigl[ {J}_{\rho}(\bm{x}(\tau )) \ln \frac{{J}_{+\rho} ({\bm x}(\tau ))}{{J}_{-\rho}({\bm x}(\tau ))} + {J}_{\rho}(\bm{x}(\tau )) \sum_{\ell \in {\cal S}} \ln  \left( x_{\ell}(\tau ) \right)^{\nabla_{\rho}^{\ell}} \Bigr] 
 \nonumber \\
&=\sigma_{m} (\tau ) + \sum_{\rho}\sum_{\ell \in {\cal S}} {\nabla_{\rho}^{\ell}}{J}_{\rho}(\bm{x}(\tau )) \ln  x_{\ell}(\tau )  \nonumber \\
&=\sigma_{m} (\tau ) + \sum_{\rho}\sum_{\ell \in {\cal S}} \dot{x}_{\ell} (\tau ) \ln  x_{\ell}(\tau ) \, . \label{prod2_lasteq}
\end{align}
At the third and fourth lines, we use the equation (\ref{macro_deterministic_eqs}). 

\subsection{Several properties due to finite concentrations}\label{sm:pro} 
In this paper, we consider the standard chemical reaction with finite concentrations, i.e., $|{\bm x(t)}|< \infty$ for $\forall t$. This assumption guarantees several properties as follows.
\begin{description}
\item[Property 1] $|J_{\pm \rho} ({\bm x}(t))|< \infty$. Long-time average, $\lim_{\tau\to\infty}{1\over\tau} \int_0^{\tau}dt J_{\pm \rho} ({\bm x}(t))$ is also finite. 
\item[Property 2] $|\sigma ({\bm x}(t))| < \infty$. Long-time average of the entropy production rate $\lim_{\tau\to\infty}{1\over \tau}\int_0^{\tau} dt \sigma ({\bm x}(t))$ is also finite. 
\item[Property 3] The gradient vector of entropy production rate is finite, i.e., $|\partial_{x_{\ell}} \sigma ({\bm x}(t))| < \infty$. 
\end{description}

\section{Fluctuation formula of the entropy production}\label{entropyfluctuation_chemLangevin}

\subsection{Size-expansion and the Langevin equations}\label{chemicalLangevineqs}
Based on the size expansion by Van Kampen and rigorous analysis by Kurtz \cite{Kurtz10.1063/1.1678692,Kurtz_1971}, we can consider the dynamics of the deviation from the deterministic motion ${\bm x}(t)$ in the form
\begin{align}
n_{\ell} (\Gamma ) /\Omega  &\to x_{\ell}(t) + y_{\ell}(t)/\sqrt{\Omega} \, , 
\end{align}
where the deterministic part $x_{\ell} (t)$ obeys the equation $\dot{x}_{\ell} = F_{\ell}({\bm x})  = \sum_{\rho}\nabla^{\ell}_{\rho} J_{\rho}({\bm x})$ and $y_{\ell} (t)/\sqrt{\Omega}$ is the fluctuating deviation part. We again emphasize that this decomposition is mathematically justified by rigorous analysis by Kurtz \cite{Kurtz10.1063/1.1678692,Kurtz_1971}. Likewise, one can also consider the fluctuating deviation from the systematic cumulative number of occurrences per volume in the form
\begin{align}
Z_{\rho} (\Gamma) / \Omega &\to z_{\rho}(t) + w_{\rho}(t)/\sqrt{\Omega} \, , 
\end{align} 
where the deterministic part $z_{\rho} (t)$ obeys the equation $\dot{z}_{\rho} = J_{\rho}({\bm x})$. 

The Langevin equations with respect to ${\bm y}(t)$ and ${\bm w}(t)$ are readily obtained by considering the Fokker-Planck equation for the variable $\bm y$ defining 
\begin{align}
f({\bm y} , t)&:={1\over \Omega ^{|{\cal S}|/2}}\sum_{{\bm n}} P({\bm n} , t) \prod_{\ell \in {\cal S}} \delta (n_{\ell} - \Omega ( x_{\ell } (t) +  y_{\ell}/\sqrt{\Omega} )) \, .
\end{align}
The time-evolution of the distribution $f({\bm y},t)$ is computed as 
\begin{align}
\partial_t f({\bm y},t) &=\sqrt{\Omega}\sum_{\ell}\dot{x}_{\ell}{\partial \over \partial y_{\ell}} f({\bm y},t) +{1\over \Omega ^{|{\cal S}|/2}}\sum_{{\bm n}} \partial_tP({\bm n} , t)  \prod_{\ell \in {\cal S}} \delta (n_{\ell} - \Omega x_{\ell } (t) - \sqrt{\Omega} y_{\ell})  \, . 
\end{align}
Performing the formal Taylor's expansion and using integration by parts several times, one arrives at the following expression in the large $\Omega$ limit:
\begin{align}
\partial_t f({\bm y} ,t) &= -  \sum_{\ell, \ell'}{\partial \over \partial y_{\ell}}\sum_{\rho}
\nabla_{\rho}^{\ell} {y}_{\ell'} {\partial \over \partial x_{\ell'}} ({J}_{+\rho } ({\bm x}(t) ) - {J}_{-\rho } ({\bm x}(t) ) )  f({\bm y} ,t) \nonumber \\
&+{1\over 2}\sum_{\ell, \ell'}{\partial^2 \over \partial y_{\ell}\partial y_{\ell'}}
\sum_{\rho}\nabla_{\rho}^{\ell}\nabla_{\rho}^{\ell'}  ({J}_{+\rho }({\bm x}(t)) +{J}_{-\rho }({\bm x}(t)) )  f({\bm y} ,t)  \, . \label{chemfp}
\end{align}
Through this Fokker-Planck equation, one can identify the corresponding Langevin equation:
\begin{align}
\dot{y}_{\ell} &= \sum_{\ell'}{\partial {F_{\ell}} ({\bm x}(t)) \over \partial x_{\ell'}} y_{\ell'} + \sum_{\rho}\nabla_{\rho}^{\ell} \sqrt{{\cal A}_{\rho} ({\bm x} (t) )} \, \xi_{\rho}(t) \, , \label{sp_yequation}\\
\dot{w}_{\rho} &= \sum_{\ell'}{\partial {J}_{\rho} ({\bm x}(t)) \over \partial {x}_{\ell'}} y_{\ell'} + \sqrt{{\cal A}_{\rho} ({\bm x} (t) )} \, {\xi}_{\rho}(t) \, . \label{sp_wequation}
\end{align} 
Here, $\xi_{\rho} (t)$ is a Gaussian white noise satisfying $\langle\!\langle \xi_{\rho} (t) \rangle\!\rangle=0$ and $\langle\!\langle \xi_{\rho} (t)\xi_{\rho'} (t')  \rangle\!\rangle=\delta_{\rho,\rho'} \delta (t-t')$, where $\langle\!\langle ... \rangle\!\rangle$ implies a noise average. The amplitude on the noise ${\cal A}_{\rho} ({\bm x} (t) )$ is given by
\begin{align}
\mathcal{A}_{\rho } (\bm{x}(t))
&:=  \left[ {J}_{+\rho}(\bm{x}(t)) + {J}_{-\rho}(\bm{x}(t)) \right] \, .
\end{align}

\subsection{Fluctuation of entropy production}\label{entropyfluc}
We consider the entropy fluctuation based on the physical environment entropy (\ref{entropybath}). We aim to derive the formal expression of fluctuation of the entropy production per unit time, defined as 
\begin{align}
{\rm Var}\sigma &:= \lim_{\tau\to\infty} {1\over \tau} \lim_{\Omega\to\infty}{1\over \Omega} \, {\rm Var} \Sigma_{\rm e} 
=\lim_{\tau\to\infty}{1\over \tau}\lim_{\Omega\to\infty} {1\over \Omega} \left[\Bigl< \left(\sum_{\rho} \mu_{\rho} Z_{\rho} (\Gamma ) \right)^2\Bigr> - \langle \sum_{\rho} \mu_{\rho} Z_{\rho} (\Gamma ) \rangle^2  \right] \, , \label{sp:varsigmaexpression}
\end{align}
where ${\rm Var} \Sigma_{\rm e}$ implies the variance of the environment entropy production. Here, note that we take the thermodynamic limit first and we next consider the long time average. The order of the limitation is crucial, since the physical properties around the bifurcation point are very sensitive on the order. Note that taking the thermodynamic first is a standard procedure to consider the phase transition behavior. 

Since we consider macroscopic-scale fluctuation of the entropy-production, we employ the Langevin equations (\ref{sp_yequation}) and (\ref{sp_wequation}). The formal solution of ${\bm y}(t)$ is readily obtained for the initial conditions ${{\bm y} (t=0)}={\bm 0}$:
\begin{align}
 {y}_{\ell} (t) &= \sum_{\rho}\sum_{\ell'}\int_{0}^{t} d s\, \Phi_{\ell, \ell'}(t, s) \nabla_{\rho}^{\ell'} \sqrt{\mathcal{A}_{\rho}(\bm{x}(s))}\, 
 {\xi}_{\rho}(s) \, ,\\
\Phi(t, s)&:= {\cal T}\exp\Bigl[
\int_{s}^{t} d s'  
{\partial {\bm F}(\bm{x}(s')) \over \partial
{\bm x} } \Bigr] \, , 
\end{align}
where ${\cal T}$ is the time-ordering operator. Using this solution, under the condition ${\bm w} (t=0)={\bm 0}$, we have 
\begin{align}
 {w}_{\rho}(\tau)  & = \sum_{\rho'}\int_{0}^\tau d s  \, G_{\rho , \, \rho'}(\tau, s ) \, \xi_{\rho' } (s) \, ,  
\end{align}
where the Green's function matrix is given by 
\begin{align}
 G_{\rho,\,\rho'}(\tau, s) &= \Bigl[ \int_{s}^{\tau} d t \sum_{\ell,\ell'}{{J}_{\rho}(\bm{x}(t))
\over \partial {x}_{\ell} } \Phi_{\ell,\ell'}(t, s) \nabla^{\ell'}_{\rho '} + \delta_{\rho,\rho'} \Bigr]\sqrt{\mathcal{A}_{\rho '}(\bm{x}(s))} 
\label{Green Func}
\end{align}
Note that the target quantity (\ref{sp:varsigmaexpression}) is obtained as 
\begin{align}
{\rm Var}\sigma &=\lim_{\tau\to\infty}{1\over \tau}\lim_{\Omega\to\infty} {1\over \Omega} \left[\Bigl< \left(\sum_{\rho} \mu_{\rho} Z_{\rho} (\Gamma ) \right)^2\Bigr> - \langle \sum_{\rho} \mu_{\rho} Z_{\rho} (\Gamma ) \rangle^2  \right] \nonumber \\
                &=\lim_{\tau\to\infty}{1\over \tau} \sum_{\rho, \rho'} \mu_{\rho} \langle\!\langle w_{\rho} (\tau) w_{\rho'} (\tau )\rangle\!\rangle \mu_{\rho'} \, ,  \label{sm:varsigmaformalformula}
\end{align}
where we use the fact that the path-probability average is replaced by the noise average in the Langevin picture. The covariance matrix with respect to the variables ${\bm w}$ is given as
\begin{align}
\lim_{\tau \to \infty}\frac{1}{\tau} \langle\!\langle w_{\rho}(\tau ) w_{\rho'}(\tau ) \rangle\!\rangle &= \lim_{\tau \to \infty}\frac{1}{\tau}\int_{0}^{\tau} d s \sum_{\rho''}G_{\rho,\, \rho''}(\tau, s) \,G_{\rho' ,\, \rho''}(\tau, s) \nonumber \\
&=\lim_{\tau \to \infty}\frac{1}{\tau}\int_{0}^{\tau} d s \sum_{\rho''}
\left[ \int_s^{\tau} dt_1 \sum_{\ell, \ell'} {\partial J_{\rho }({\bm x}(t_1) ) \over \partial x_{\ell}}\Phi_{\ell,\ell'} (t_1 , s) \nabla^{\ell'}_{\rho ''} + \delta_{\rho,\, \rho''} \right]
\nonumber \\
&~~~~~~~~~~~\times\mathcal{A}_{\rho''}(\bm{x}(s))  \left[ \int_s^{\tau} dt_2 \sum_{m,m'} {\partial J_{\rho'}({\bm x}(t_2) ) \over \partial x_{m}}\Phi_{m,m'} (t_2 , s) \nabla^{m'}_{\rho''} +  \delta_{\rho', \,\rho''}\right] \, . 
\end{align}
After manipulation, we arrive at the following formula when the stable steady solution of Eq.  (\ref{macro_deterministic_eqs}) is a fixed point such as in pitchfork, transcritical, and saddle-node bifurcations:
\begin{align}
\begin{split}
{\rm Var}\sigma &= \sum_{\rho\in {\cal R}} v_{\rho}^2(\bar{\bm x}) {\cal A}_{\rho}(\bar{\bm x})\, , ~~~~~~~~~~~~~~~~[{\rm Fixed~point~case}]  \label{sm:fpformula}\\ 
v_{\rho} &= \mu_{\rho} -\sum_{\ell, \ell' \in {\cal S}} {\partial_{\bar{x}_{\ell}} \sigma (\bar{\bm x})} [{\mathbb S}^{-1} (\bar{\bm x} )]_{\ell, \ell'} \nabla_{\rho}^{\ell '} \, , 
\end{split}
\end{align}
where $\bar{\bm{x}}$ is the stable fixed point solution in Eq.(\ref{macro_deterministic_eqs}), and $\mathbb{S} (\bar{\bm{x}})$ is the stability matrix (Jacobian) evaluated at the fixed point, with elements $[\mathbb{S} (\bar{\bm{x}})]_{\ell, \ell'} = \partial F_{\ell} / \partial x_{\ell'} |_{\bm{x}=\bar{\bm{x}}}$. 
Conversely, when the stable steady solution of (\ref{macro_deterministic_eqs}) is a limit cycle (LC), typically arising from a Hopf bifurcation, the formula is given by:
\begin{align}
\begin{split}
{\rm Var}\sigma &= \! \lim_{\tau\to\infty}\!{1\over \tau}\! \int_{0}^{\tau} \!\! ds \! \sum_{\rho}\! v_{\rho}^2(\bar{\bm x}(s)) {\cal A}_{\rho}(\bar{\bm x}(s))\,,~~~~~~~[{\rm Limit~cycle~case}] \, , \label{sm:lcformula}\\ 
v_{\rho} (\bar{\bm x}(s)) &=\!\mu_{\rho} + \!\!\sum_{\ell, \ell' } \!\int_s^{\tau} \!\! dt \partial_{\bar{x}_{\ell}} \sigma(\bar{\bm x}(s)) \Phi_{\ell, \ell'} (t,s) \nabla_{\rho}^{\ell'} \, ,
\end{split} 
\end{align}
where $\bar{\bm x}(s)$ is a periodic solution (limit cycle) of Eq.(\ref{macro_deterministic_eqs}). 

Note that the term $ {\cal A}_{\rho}$ never diverges due to the Property 1 in Sec.\ref{sm:pro}. Divergence of the fluctuation emerges from the term $v_{\rho}$. 

\section{Classification of diverging exponents in entropy fluctuation}\label{classification_of_entropyfluctuation}
In the chemical reaction, we consider the control parameter $\theta$ defined through the 
concentration of chemostatted species:
\begin{align}
{\bm a} &= {\bm a}^{(0)}+\theta \hat{\bm a} \, , \label{sm:defoftheta}
\end{align}
where $\hat{\bm a}=(\hat{a}_{\ell})_{\ell\in {\cal S}_c}$ is the unit vector which defines the increasing direction of the concentrations of the chemostatted chemical species from the initial concentration vector ${\bm a}^{(0)}$. The parameter $\theta$ is the amplitude along the direction. The parameter $\theta$ is a control parameter to govern the bifurcations. The final aim here is to classify the exponent of the divergence of the fluctuation of the entropy production according to the types of bifurcation:
\begin{align}
{\rm Var}\sigma &\propto
\left\{
\begin{array}{ll}
(\theta - \theta_c)^{-\alpha_+} & \theta > \theta_c \\
(\theta_c - \theta )^{-\alpha_-} & \theta < \theta_c \\
\end{array} 
\right.
\end{align}
Note that exponents $\alpha_{\pm}$ does not depend on the choice of the functional form on the parameter $\theta$, as long as the function is differentiable around the bifurcation point.

\subsection{Typical bifurcations and their normal forms: A review}\label{sm:bifurcationnormalform}
Before classifying the diverging exponents at bifurcation points, we briefly review the local dynamical structure near several standard bifurcations. Let us consider an $n$-dimensional state vector ${\bm x}$ obeying the deterministic dynamics
\[
\dot{\bm x}={\bm F}_{\theta}({\bm x}) \, .
\]
The subscript $\theta$ emphasizes that the vector field depends on a control parameter $\theta$. The linearized dynamics around a solution ${\bm x}(t)$ is governed by the stability matrix $\mathbb{S}_{\theta}({\bm x}(t))$:
\begin{align}
\delta \dot{\bm x}(t)
&=
\mathbb{S}_{\theta}({\bm x}(t)) \, \delta {\bm x}(t) \, ,
\qquad
\mathbb{S}_{\theta}({\bm x}(t))
:=
\left.
\frac{\partial {\bm F}_{\theta}}{\partial {\bm x}}
\right|_{{\bm x}={\bm x}(t)} .
\label{sm:jacobiandef}
\end{align}
Let $\theta_c$ be a bifurcation point, and let ${\bm x}_c$ be an equilibrium (stable fixed point solution) at $\theta=\theta_c$, namely
\[
{\bm F}_{\theta_c}({\bm x}_c)={\bm 0}.
\]
The eigenvalues of the linearization ${\mathbb S}_{\theta_c}({\bm x}_c)$ determine the decomposition of the tangent space
\begin{align}
T_{{\bm x}_c}\mathbb{R}^n
&=
E^{\rm c}
\oplus
E^{\rm s}
\oplus
E^{\rm u} \, ,
\end{align}
where $E^{\rm c}$ is the center subspace associated with eigenvalues of zero real part, and $E^{\rm s}$ and $E^{\rm u}$ denote the stable and unstable subspaces, respectively. For $\theta$ close to $\theta_c$, the local dynamics relevant to the bifurcation is governed by the flow restricted to the center manifold. The {\it Center Manifold Theorem} \cite{Kuznetsov2023Elements} guarantees the existence of a locally invariant center manifold tangent to $E^{\rm c}$ at ${\bm x}_c$.

Let ${\bm e}_1^{\rm c}, \cdots , {\bm e}_{d_{\rm c}}^{\rm c}$ be basis vectors for the center subspace $E^{\rm c}$, and let ${\bm e}^{\rm su}_1, \cdots , {\bm e}^{\rm su}_{d_{\rm su}}$ be basis vectors for the stable and unstable subspaces $E^{\rm s}\oplus E^{\rm u}$. Here, $d_{\rm c}$ and $d_{\rm su}$ are the dimensions of $E^{\rm c}$ and $E^{\rm s}\oplus E^{\rm u}$, respectively, and hence $d_{\rm c}+d_{\rm su}=n$. These vectors are not necessarily orthogonal. We construct the matrix ${\cal P}$ by taking these vectors as its columns:
\begin{align}
{\cal P}
&=
({\bm e}_1^{\rm c}, \cdots , {\bm e}_{d_{\rm c}}^{\rm c}, {\bm e}^{\rm su}_1, \cdots , {\bm e}^{\rm su}_{d_{\rm su}}) \, .
\end{align}
Then, for $\theta$ near $\theta_c$, we introduce linear coordinates $({\bm c}^T,{\bm s}^T)^T$ by
\begin{align}
{\bm x}-{\bm x}_c
&=
{\cal P}
\left(
\begin{array}{c}
{\bm c}\\
{\bm s}
\end{array}
\right) \, ,
\label{sm:xandzeta}
\end{align}
where ${\bm c}=(c_1,\cdots,c_{d_{\rm c}})^T$ parameterizes the center directions, while ${\bm s}=(s_1,\cdots,s_{d_{\rm su}})^T$ parameterizes the stable and unstable directions.

Introducing the unfolding parameter
\[
\mu := \theta-\theta_c \, ,
\]
the {\it Center Manifold Theorem} guarantees that, for $|\mu|\ll 1$, the center manifold can be represented locally as a graph
\[
{\bm s}={\bm h}({\bm c},\mu)
\]
with ${\bm h}({\bm 0},0)={\bm 0}$ and $\partial_{\bm c}{\bm h}({\bm 0},0)={\bm 0}$. Accordingly, the variables ${\bm c}$ and ${\bm s}$ obey equations of the form
\begin{align}
\dot{\bm c} &= {\bm g}_{\rm c}({\bm c},{\bm s},\mu) \, ,\\
\dot{\bm s} &= {\bm g}_{\rm su}({\bm c},{\bm s},\mu) \, ,
\end{align}
and the reduced dynamics on the center manifold is
\[
\dot{\bm c}={\bm g}_{\rm c}({\bm c},{\bm h}({\bm c},\mu),\mu) \, .
\]

To flatten the center manifold, we introduce the nonlinear coordinate transformation
\[
{\bm w}={\bm s}-{\bm h}({\bm c},\mu) \, .
\]
With this transformation, the center manifold is given by ${\bm w}={\bm 0}$. 
When the system starts at the initial condition $({\bm c}(0),{\bm w}(0)={\bm 0})$, then the state at time $t$ remains in the manifold $({\bm c}(t),{\bm w}(t)={\bm 0})$ for any time $t$. We define the flattened coordinates by ${\bm \zeta}=({\bm c}^T,{\bm w}^T)^T$. Then
\[
{\bm x}
=
{\bm x}_c
+
{\cal P}
\left(
\begin{array}{c}
{\bm c}\\
{\bm w}+{\bm h}({\bm c},\mu)
\end{array}
\right),
\]
and the Jacobian matrix of this coordinate map is
\begin{align}
{\cal J}({\bm c},\mu)
&:=
\frac{\partial {\bm x}}{\partial {\bm \zeta}}
=
{\cal P}
\left(
\begin{array}{cc}
{\mathbb I} & {\bm 0}\\
\partial_{\bm c}{\bm h}({\bm c},\mu) & {\mathbb I}
\end{array}
\right) \, .
\end{align}
At $({\bm c},\mu)=({\bm 0},0)$, it coincides with ${\cal P}$. Let $\tilde{\bm F}_{\theta}({\bm \zeta})$ denote the transformed vector field in the flattened coordinates satisfying  
\begin{align}
\dot{\bm \zeta}
=\tilde{\bm F}_{\theta}({\bm c},{\bm w}) \, . \label{sm:flattenddetdynam}
\end{align}
We also define its stability matrix as
\begin{align}
\tilde{\mathbb S}_{\theta}({\bm \zeta}) 
=\tilde{\mathbb S}_{\theta}({\bm c},{\bm w}) 
&:=
\partial_{\bm \zeta}\tilde{\bm F}_{\theta}({\bm c}, {\bm w}) \, . 
\end{align}
 Since the set ${\bm w}={\bm 0}$ is invariant, the lower component of the transformed vector field satisfies $\tilde{\bm F}_{\theta}^{({\rm su})}({\bm c},{\bm 0})={\bm 0}$ for
$\forall {\bm c}$. Therefore, the stability matrix along the center manifold has the block upper-triangular form
\begin{align}
\tilde{\mathbb S}_{\theta}({\bm c},{\bm w}={\bm 0})
&=
\left(
\begin{array}{cc}
\tilde{\mathbb S}_{\theta}^{({\rm c})}({\bm c})
&
\tilde{\mathbb S}_{\theta}^{({\rm csu})}({\bm c})
\\
{\bm 0}
&
\tilde{\mathbb S}_{\theta}^{({\rm su})}({\bm c})
\end{array}
\right) \, .
\end{align}
Moreover, if
\[
\bar{\bm \zeta}=
\left(
\begin{array}{c}
\bar{\bm c}\\
{\bm 0}
\end{array}
\right)
\]
is an equilibrium of the dynamics (\ref{sm:flattenddetdynam}), then the corresponding equilibrium of the original system is
\[
\bar{\bm x}
=
{\bm x}_c
+
{\cal P}
\left(
\begin{array}{c}
\bar{\bm c}\\
{\bm h}(\bar{\bm c},\mu)
\end{array}
\right),
\]
and the two stability matrices are related to each other as
\begin{align}
{\mathbb S}_{\theta}(\bar{\bm x})
&=
{\cal J}(\bar{\bm c},\mu)\,
\tilde{\mathbb S}_{\theta}(\bar{\bm c},{\bm w}={\bm 0})\,
{\cal J}(\bar{\bm c},\mu)^{-1} \, .
\end{align}
Hence, the eigenvalues at equilibria on the center manifold are determined by the diagonal blocks $\tilde{\mathbb S}_{\theta}^{({\rm c})}$ and $\tilde{\mathbb S}_{\theta}^{({\rm su})}$. The eigenvalues associated with $E^{\rm s}\oplus E^{\rm u}$ remain $O(1)$ and stay away from the imaginary axis near the bifurcation.

In the case of $d_{\rm c}=1$, the reduced dynamics on the center manifold is one-dimensional. Then ${\bm c}$ becomes a scalar $c$, and the eigenvalue associated with the center direction converges to $0$ as $\theta\to\theta_c$. Expanding the reduced vector field in powers of $c$ and $\mu$, we obtain
\[
\dot{c}
=
g_{\rm c}(c,\mu)
=
a_{10}\mu
+
a_{01}c
+
a_{11}\mu c
+
a_{02}c^2
+
a_{03}c^3
+
\cdots .
\]
At the bifurcation point, $a_{01}=\partial_c g_{\rm c}(0,0)=0$. The bifurcation type is determined by the lowest-order nonvanishing nonlinear term together with the nondegeneracy conditions with respect to the unfolding parameter. After smooth coordinate changes and rescalings of $c$, $\mu$, and time, one obtains the standard normal forms. \\
\noindent
If an additional $\mathbb{Z}_2$ symmetry $c\mapsto -c$ is present, then the reduced vector field is odd in $c$, namely $g_{\rm c}(-c,\mu)=-g_{\rm c}(c,\mu)$, so the constant and even-order terms vanish. If, moreover, the coefficients of $\mu c$ and $c^3$ are nonzero, a pitchfork bifurcation occurs. For definiteness, the supercritical normal form is written as
\begin{align}
\dot{c}
&=
g_{\rm c}(c,\mu)
=
\mu c-c^3 .
\label{sm:pitchfork}
\end{align}
If $a_{10}=0$, $a_{11}\neq 0$, and $a_{02}\neq 0$, a transcritical bifurcation occurs:
\begin{align}
\dot{c}
&=
g_{\rm c}(c,\mu)
=
\mu c-c^2 .
\label{sm:transcritical}
\end{align}
If $a_{10}\neq 0$ and $a_{02}\neq 0$, a saddle-node bifurcation occurs:
\begin{align}
\dot{c}
&=
g_{\rm c}(c,\mu)
=
\mu-c^2 .
\label{sm:saddlenode}
\end{align}

In the case of $d_{\rm c}=2$, the center block of the linearization has a pair of complex conjugate eigenvalues
\begin{align}
\lambda_{1,2}^{\rm c}(\theta)
&=
\alpha(\theta)
\pm i\omega(\theta) \, ,
\label{sm:hopfeigen}
\end{align}
such that $\alpha(\theta_c)=0$, $\omega(\theta_c)=\omega_0\neq 0$, and $\frac{d\alpha}{d\theta}(\theta_c)\neq 0$, while all other eigenvalues have negative real parts. By the Center Manifold Theorem, the dynamics reduces locally to a two-dimensional system. After smooth coordinate transformations and rescalings, the reduced system can be written in complex form as
\[
\dot{z}
=
(\mu+i\omega_0)z
-
\gamma z|z|^2
+ \cdots \, , 
\]
where $z=c_1+ic_2\in\mathbb{C}$, $\mu=\theta-\theta_c$, and $\gamma\in\mathbb{R}$ determines the criticality (for simplicity, we suppress the nonlinear frequency shift). Truncating at cubic order, we obtain the Hopf normal form in the two variables $c_1$ and $c_2$:
\begin{equation}
\begin{split}
\dot{c}_1
&=
g_{{\rm c},1}(c_1,c_2)
=
\mu c_1
-
\omega_0 c_2
-
\gamma c_1(c_1^2+c_2^2) + \cdots \, , \\
\dot{c}_2
&=
g_{{\rm c},2}(c_1,c_2)
=
\omega_0 c_1
+
\mu c_2
-
\gamma c_2(c_1^2+c_2^2) + \cdots \, .
\end{split}
\label{sm:hopf}
\end{equation}
Writing $z=re^{i\phi}$, we obtain the equations
\begin{align}
\dot{r}
&=
\mu r-\gamma r^3 + O(\mu^2 r  + \mu r^3 + r^5),
\qquad
\dot{\phi}
=
\omega_0 + O(\mu + r^2).
\end{align}
If $\gamma>0$, the Hopf bifurcation is supercritical, and a stable limit cycle of radius emerges for $\mu>0$ (equivalently, $\theta>\theta_c$):
\begin{align}
r&=\sqrt{\mu/\gamma}\, , ~~~T(\mu) = {2\pi \over \omega_0} + O(\mu) \, , \label{sm:randT}
\end{align}
where $T(\mu)$ is the period of the limit cycle. If $\gamma<0$, the Hopf bifurcation is subcritical, and an unstable limit cycle exists for $\mu<0$ (equivalently, $\theta<\theta_c$).

\subsection{Fluctuation of the entropy production for fixed-point cases}\label{fixpoint}
Let us discuss the generic properties for the case $d_{\rm c}=1$, where the center manifold is one-dimensional. In this case, the relevant stable states are equilibria, and hence we use the formula for the fixed-point case in (\ref{sm:fpformula}). We write
$\mu := \theta-\theta_c$ for the unfolding parameter. Let $\bar{\bm x}_{\pm}(\theta)$ denote a stable equilibrium approached from the sides $\theta\to\theta_c\pm 0$, respectively, whenever the corresponding limit exists. We note each element of the entropy gradient vector $\frac{\partial \sigma}{\partial x_{\ell}}$ never diverges due to the property 3 in Sec.\ref{sm:pro}. We consider the {\it generic} situation where the entropy-gradient vector has a finite, nonzero overlap with the center direction. 
In the generic case, the divergence of the fluctuation of entropy production is controlled by the singular part of the inverse stability matrix along the center-manifold equilibrium. Let
\[
\bar{\bm \zeta}=
\left(
\begin{array}{c}
\bar{c} \\
0
\end{array}
\right)
\]
be a stable equilibrium in the flattened coordinates, and let $\bar{\bm x}$ be the corresponding equilibrium in the original coordinates. Then, from the block upper-triangular structure of the transformed stability matrix on the center manifold, we obtain
\begin{align}
[{\mathbb S}_{\theta}(\bar{\bm x})]^{-1}
&=
{\cal J}(\bar{c},\mu)
\tilde{\mathbb S}_{\theta}(\bar{\bm \zeta})^{-1}
[{\cal J}(\bar{c},\mu)]^{-1}
\nonumber \\
&=
{\cal J}(\bar{c},\mu)
\left(
\begin{array}{cc}
\lambda^{-1}
&
-\lambda^{-1}
\tilde{\mathbb S}_{\theta}^{({\rm csu})}(\bar{c})
[ \tilde{\mathbb S}_{\theta}^{({\rm su})}(\bar{c}) ]^{-1}
\\
{\bm 0}
&
[ \tilde{\mathbb S}_{\theta}^{({\rm su})}(\bar{c}) ]^{-1}
\end{array}
\right)
[{\cal J}(\bar{c},\mu)]^{-1} \, ,
\end{align}
where
\[
\lambda
:=
\tilde{\mathbb S}_{\theta}^{({\rm c})}(\bar{c},\mu)
=
\partial_c g_{\rm c}(\bar{c},\mu)
\]
is the center eigenvalue at the equilibrium. Near the bifurcation point, the diverging contribution comes from $\lambda^{-1}$ because $\lambda\to 0$ as $\theta\to\theta_c$, while the elements of $\tilde{\mathbb S}_{\theta}^{({\rm csu})}$ and the eigenvalues of $\tilde{\mathbb S}_{\theta}^{({\rm su})}$ remain $O(1)$. Let $\lambda_{+}$ and $\lambda_{-}$ denote the values of $\lambda$ on the stable branches for $\mu>0$ and $\mu<0$, respectively. In the generic case described above, the fluctuation of entropy production scales as
\[
\mathrm{Var}\,\sigma \propto \lambda_{\pm}^{-2}
\]
on each side where a stable equilibrium branch exists. We now classify this generic behavior for each bifurcation class.

\subsubsection{Pitchfork and transcritical bifurcations}
For the pitchfork bifurcation, we use the normal form (\ref{sm:pitchfork}),
\[
\dot c = \mu c-c^3 \, .
\]
The stable equilibria are given by $\bar{c}=\pm\sqrt{\mu}$ for $\mu>0$ and $\bar{c}=0$ for $\mu<0$. Hence,
\begin{align}
\lambda_{+}
&=
\partial_c g_{\rm c}(\bar{c},\mu)
=
\mu-3\bar{c}^2
=
-2\mu
\qquad (\mu>0) \, ,
\\
\lambda_{-}
&=
\partial_c g_{\rm c}(0,\mu)
=
\mu
=
-|\mu|
\qquad (\mu<0) \, .
\end{align}
Consequently, the generic behavior of the fluctuation of entropy production is
\begin{align}
\mathrm{Var}\,\sigma
&\propto
\begin{cases}
\dfrac{C_{+}}{\mu^2}
& \cdots \ \mu>0 \, ,\\[1.2ex]
\dfrac{C_{-}}{\mu^2}
& \cdots \ \mu<0 \, ,
\end{cases}
\label{sm:var_pitchfork}
\end{align}
where $C_{\pm}$ are finite constants. Thus, the exponents are $(\alpha_+,\alpha_-)=(2,2)$ for the pitchfork bifurcation. We have some remark on the symmtery in the pictchfork bifurcation. 
The normal form has $c\to -c$ symmetry. However, we emphasize that this symmetry does not imply $\sigma(c)=\sigma(-c)$. The entropy production is a physical observable expressed in the original concentrations and reservoir affinities, and it generally contains an odd component in the center-manifold coordinate, $\sigma(c)=\sigma_0+s_1c+s_2c^2+\cdots$ with $s_1\neq0$. Only when the full chemical reaction network has an additional exact symmetry that exchanges the two branches and leaves the environment entropy production invariant is the odd coefficient $s_1$ forbidden. Such symmetry-protected cases are nongeneric in the present sense and may exhibit reduced exponents. 

We next consider the transcritical bifurcation in the same manner. We use the normal form (\ref{sm:transcritical}),
\[
\dot c = \mu c-c^2 \, ,
\]
which has the stable equilibria $\bar{c}=\mu$ for $\mu>0$ and $\bar{c}=0$ for $\mu<0$. Therefore,
\begin{align}
\lambda_{+}
&=
\partial_c g_{\rm c}(\mu,\mu)
=
\mu-2\mu
=
-\mu
\qquad (\mu>0) \, ,
\\
\lambda_{-}
&=
\partial_c g_{\rm c}(0,\mu)
=
\mu
=
-|\mu|
\qquad (\mu<0) \, .
\end{align}
Hence, we arrive at the same type of divergence for the fluctuation of entropy production as in (\ref{sm:var_pitchfork}). Therefore, the exponents are again $(\alpha_+,\alpha_-)=(2,2)$ for the transcritical bifurcation.

The pitchfork and transcritical bifurcations are both continuous transitions in the sense that the stable equilibria change continuously as a function of $\theta$. In this generic setting, the continuous bifurcation cases show the same exponents on both sides, namely $(\alpha_{+},\alpha_{-})=(2,2)$.

\subsubsection{Saddle-node bifurcation}
In contrast to the pitchfork and transcritical bifurcations, the saddle-node bifurcation shows a sudden loss of the local stable equilibrium branch. Using the normal form (\ref{sm:saddlenode}),
\[
\dot c = \mu-c^2 \, ,
\]
one finds the stable equilibrium $\bar{c}=\sqrt{\mu}$ for $\mu>0$. Therefore,
\begin{align}
\lambda_{+}
&=
\partial_c g_{\rm c}(\bar{c},\mu)
=
-2\bar{c}
=
-2\sqrt{\mu}
\qquad (\mu>0) \, .
\end{align}
Hence, the generic divergence of the entropy-production fluctuation on the side where the stable equilibrium exists is
\begin{align}
\mathrm{Var}\,\sigma
&\propto
\frac{C''_{+}}{\mu}
\qquad (\mu>0) \, ,
\end{align}
where $C''_{+}$ is a finite constant.

For $\mu<0$, the local normal form has no nearby equilibrium, so the center-manifold analysis does not provide a universal singular contribution on that side. Any behavior for $\theta<\theta_c$ is determined by the global structure of the full system.
Therefore, from the local bifurcation analysis, one concludes the generic exponents $(\alpha_+,\alpha_-)=(1,\varnothing)$ for the saddle-node bifurcation. Here, $\varnothing$ implies that the exponent depends on the specific models. 

The saddle-node bifurcation is discontinuous in the sense that the local stable equilibrium branch terminates at the bifurcation point. This is reminiscent of a first-order equilibrium phase transition. However, unlike the equilibrium first-order case, the bifurcation can still show one-sided scaling behavior through the vanishing center eigenvalue. We also remark on the difference from the argument in Ref.~\cite{NguyenPhysRevE.102.022101Schlogl}, where exponentially large fluctuations are observed for finite-volume chemical reactions. As emphasized before, we focus only on the thermodynamic limit (\ref{sp:varsigmaexpression}). Hence, our target quantity is different from the one considered in Ref.~\cite{NguyenPhysRevE.102.022101Schlogl}.


\subsection{Fluctuation of the entropy production for limit-cycle case}\label{hopf}
In the Hopf bifurcation case, we consider the formula in (\ref{sm:lcformula}) for 
time-periodic limit cycle in the regime $\theta > \theta_c$, while we use (\ref{sm:fpformula}) for the regime $\theta < \theta_c$.

\subsubsection{The case of $\theta > \theta_c$}
We here discuss the divergence behavior for the regime $\mu:=\theta - \theta_c>0$ near the Hopf bifurcation point. Since the periodic motion appears in this regime, we use the monodromy matrix ${\bm M}_{\theta}$ defined as
\begin{align}
{\bm M}_{\theta} &= {\cal T} \exp \Bigl[\int_0^T dt \, {\mathbb S}_{\theta} (\bar{\bm x} (t) ) \Bigr] \, , 
\end{align}
where $\bar{\bm x} (t)$ is the stable periodic solution of Eq.(\ref{macro_deterministic_eqs}). $T$ is the period of the limit cycle which depends on the parameter $\theta$ as
\begin{align}
T &= {2\pi \over \omega_0} + O(\mu) \, . \label{sm:period}
\end{align}
See Eq.(\ref{sm:randT}) in the review part of the bifurcations in Sec.\ref{sm:bifurcationnormalform}. We define the left and right eigenvectors for the monodromy matrix ${\bm M}_{\theta}$:
\begin{align}
    {\bm M}_{\theta} \bm{v}_j & = e^{\nu_j T} \bm{v}_j \, , \quad \bm{w}^\dagger_j {\bm M}_{\theta} = e^{\nu_j T} \bm{w}^\dagger_j \, .
\end{align} 
We can show the following lemma for the monodromy matrix, the proof of which is provided in the sec. \ref{sm:proofsection} below.  
\begin{lemma}\label{lemma1}
Assume a generic supercritical Hopf bifurcation at $\theta=\theta_c$, and let
$\mu=\theta-\theta_c>0$.
Then the Floquet exponents of the limit cycle born from the Hopf bifurcation satisfy
\begin{align}
\begin{split}
\nu_1&=0,\qquad \nu_2=-2\mu+O(\mu^2), \\
\nu_j&=\lambda_j(\theta_c)+O(\mu),\qquad j=3,\dots,n,
\end{split}
\end{align}
where $\Re \lambda_j(\theta_c)\neq 0$ for $j\ge 3$.
Namely, one exponent is zero, one is of order $O(\theta-\theta_c)$,
and the remaining $n-2$ exponents are of order $O(1)$.
\end{lemma}

The properties $\bm{w}_j^\dagger \bm{v}_m = \delta_{jm}$ and $\sum_{j}  \bm{v}_j \bm{w}^\dagger_j  = {\mathbb I}$ are satisfied. Then, we can write $\Phi_{\theta} (t,0) = \sum_{j} e^{\nu_j t} \tilde{\bm{v}}_j(t) \bm{w}^\dagger_j$, where $\tilde{\bm{v}}_j(t) := e^{- \nu_j t} \Phi_{\theta}(t, 0) \bm{v}_j$. To correctly decompose the matrix as $\Phi_{\theta}(t,s) = \Phi_{\theta}(t,0)\Phi_{\theta}(0,s)$, we define the periodic left eigenvectors as $\tilde{\bm{w}}^\dagger_j(t) := e^{\nu_j t} \bm{w}^\dagger_j \Phi_{\theta}(0,t)$. Note that $\tilde{\bm{v}}_j(t)$ and $\tilde{\bm{w}}^\dagger_j(t)$ are both periodic functions with the period $T$. One can then write the matrix $\Phi_{\theta}(t, s)$ as
\begin{align}
    \Phi_{\theta}(t, s) & =  \sum_{j=1}^n \Phi_{\theta,j}(t, s) \,, \quad \Phi_{\theta,j} (t,s) = e^{\nu_j (t - s)} \tilde{\bm{v}}_j(t) \tilde{\bm{w}}^\dagger_j(s) \, .
\end{align}


In order to investigate ${\rm Var}\sigma$, let us focus on the term $v_{\rho} - \mu_{\rho}$ in the formula (\ref{sm:lcformula}):
\begin{align}
 v_{\rho}(\bar{\bm x}( s)) - \mu_{\rho} & \equiv \int_{s}^{\tau} d t \sum_{\ell, \ell'}{\partial \sigma (\bm{x} (t)) \over \partial x_{\ell}} [\Phi_{\theta} (t, s)]_{\ell,\ell'} \nabla_{\rho}^{\ell'} =\sum_{j=1}^n [g_j (\tau, s)]_{\rho}
\, , \\
[g_j (\tau, s)]_{\rho} &= \int_{s}^{\tau} d t\sum_{\ell,\ell'} {\partial \sigma (\bm{x} (t)) \over \partial x_{\ell}} [\Phi_{\theta,j} (t, s)]_{\ell,\ell'} \nabla_{\rho}^{\ell'}  \, .
\end{align}

We first consider the term $g_{j=1}$. Note that $\nu_1=0$ and $\tilde{\bm v}_1 (t) \propto \dot{\bm x}(t)$. Hence we observe
\begin{align}
[g_1 (\tau, s)]_{\rho} &\propto \int_{s}^{\tau} d t \sum_{\ell,\ell'} {\partial \sigma(\bm{x} (t)) \over \partial x_{\ell }} \dot{x}_{\ell} (t) [\tilde{\bm{w}}^\dagger_1(s)]_{\ell'} \nabla_{\rho}^{\ell'} = \left[ \sigma (\bm{x} (\tau ))  - \sigma (\bm{x} (s )) \right] \sum_{\ell' }[\tilde{\bm{w}}^\dagger_1(s)]_{\ell'} \nabla_{\rho}^{\ell'} \, . 
\end{align}
From the property 2 in Sec.\ref{sm:pro}, this term never contributes to the divergence. Hence, we next consider the other cases $j\ge 2$. To this end, we write the term dependent on $t$ in the Fourier transform:
\begin{align}
\sum_{\ell} {\partial \sigma ({\bm x}(t)) \over \partial  x_{\ell} } [\tilde{\bm v}_j (t)]_{\ell} &= \sum_{k=-\infty}^{\infty} {\cal F}_{j,k} (\theta )  e^{i \omega k t} \, , \quad \omega=2\pi/T \, ,  
\end{align}
where we explicitly write the $\theta$-dependence in the Fourier transform. Then, we have
\begin{align} 
[g_j (\tau, s)]_{\rho} &=\sum_{k=-\infty}^{\infty} [g_{j,k} (\tau, s)]_{\rho} \, , \\
[g_{j,k} (\tau, s)]_{\rho} &= { {\cal F}_{j,k} (\theta ) \over \nu_j + i \omega k}( e^{i \omega k \tau} e^{\nu_j (\tau - s)} - e^{i \omega k s} ) \sum_{\ell'} [\tilde{\bm{w}}^\dagger_j(s)]_{\ell'} \nabla_{\rho}^{\ell'}  \, .
\end{align}
Note that $\omega$ is finite at and near the bifurcation point, and hence this structure tells us that the component $k=0$ and $j=2$ gives singular behavior, since $\nu_2 = O(\mu)$ while $\nu_j =O(1)$ for $j\ge 3$. Hence, the fluctuation is dominated by $[{\rm Var}\sigma]\sim [{\rm Var}\sigma]_{\rm D} $, where $[{\rm Var}\sigma]_{\rm D}$ is given by the component $g_{j=2,k=0}$ as
\begin{align}
  [{\rm Var}\sigma]_{\rm D} &:= \lim_{\tau\to\infty} {1\over \tau} \int_0^{\tau} ds 
\sum_{\rho} |[g_{2,0} (\tau, s)]_{\rho} + \mu_\rho |^2 {\cal A}_{\rho}(\bar{\bm x}(s))
=  [{\rm Var}\sigma]_{\rm D}^{(1)}  +  [{\rm Var}\sigma]_{\rm D}^{(2)} + {\rm const}.\, ,  \label{sm:varsigmad} \\
[{\rm Var}\sigma]_{\rm D}^{(1)} &=\lim_{\tau\to\infty}{1\over \tau} \int_{0}^{\tau} ds \sum_{\rho} [g_{2,0} (\tau, s)]_{\rho}[g_{2,0} (\tau, s)]_{\rho}^{\ast} {\cal A}_{\rho}(\bar{\bm x}(s)) \, , \\
[{\rm Var}\sigma]_{\rm D}^{(2)} &=\lim_{\tau\to\infty}{1\over \tau} \sum_{\rho}\mu_{\rho}{\rm Re} \int_{0}^{\tau} ds [g_{2,0} (\tau, s)]_{\rho} {\cal A}_{\rho}(\bar{\bm x}(s)) \, .
\end{align}
where the subscript `D' stands for the dominant contribution of the fluctuation. 

We first consider the term $[{\rm Var}\sigma]_{\rm D}^{(1)}$, which is explicitly written as 
\begin{align}
[{\rm Var}\sigma]_{\rm D}^{(1)} &= \lim_{\tau\to\infty}{1\over \tau} \int_{0}^{\tau} ds \sum_{\rho} \sum_{\ell,\ell'} { |{\cal F}_{2,0} (\theta ) |^2 \over |\nu_2|^2} |e^{\nu_2 (\tau -s)} - 1 |^2 [\tilde{\bm{w}}^\dagger_2(s)]_{\ell} \nabla_{\rho}^{\ell} {\cal A}_{\rho}(\bar{\bm x}(s))[\tilde{\bm{w}}^\dagger_2(s)]_{\ell'}^{\ast} \nabla_{\rho}^{\ell'} \, .
\end{align}
Let use the Fourier transform as follows
\begin{align}
\sum_{\rho,\ell,\ell'}
[\tilde{\bm{w}}^\dagger_2(s)]_{\ell} \nabla_{\rho}^{\ell} {\cal A}_{\rho}(\bar{\bm x}(s))[\tilde{\bm{w}}^\dagger_2(s)]_{\ell'}^{\ast} \nabla_{\rho}^{\ell'} &= \sum_{k'=-\infty}^{\infty} {\cal D}_{k'}(\theta ) e^{i \omega k' s} \, .
\end{align}
Then, we have 
\begin{align}
[{\rm Var}\sigma]_{\rm D}^{(1)} &=\sum_{k'=-\infty}^{\infty} \lim_{\tau\to\infty}{1\over \tau} \int_{0}^{\tau} ds { |{\cal F}_{2,0} (\theta ) |^2 \over |\nu_2|^2} |e^{\nu_2 (\tau -s)} - 1 |^2{\cal D}_{k'}(\theta ) e^{i \omega k' s}  \nonumber \\
&={ |{\cal F}_{2,0} (\theta ) |^2 \over |\nu_2|^2} {\cal D}_{0}(\theta ) \, , 
\end{align}
by using the fact that $\lim_{\tau\to\infty}{1/\tau}\int_0^{\tau} ds |e^{\nu_2 (\tau -s)} - 1 |^2 e^{i \omega k' s} =\delta_{k',0}$. Note that the term ${\cal F}_{2,0}(\theta)$ and ${\cal D}_0 (\theta )$ are written as 
\begin{align}
{\cal F}_{2,0} (\theta ) &= \int_0^{T} dt  \sum_{\ell}{\partial \sigma  (\bar{\bm x}(t)) \over \partial x_{\ell}} [\tilde{\bm v}_2 (t)]_{\ell} \, , \\
{\cal D}_0 (\theta ) &= \int_0^T d t\, \sum_{\rho ,\ell,\,\ell'}[\tilde{\bm w}_2^{\dagger} (t)]_{\ell} \nabla_{\rho}^{\ell} {\cal A}_{\rho} (\bar{\bm x} (t)) \nabla_{\rho}^{\ell'} [\tilde{\bm w}_2^{\dagger} (t)]_{\ell'}^{\ast} \, .
\end{align}
The term ${\cal D}_0 (\theta)$ neither diverges due to $|{\cal A}_{\rho} (\bar{\bm x} (t))|<\infty$ from the property 1 in Sec.\ref{sm:pro}, nor is infinitesimal since $\theta\to\theta_c$ leads to a finite value. On the other hand, as we show below, ${\cal F}_{2,0} (\theta ) = {O} (\sqrt{\mu})$.
To show this, we expand the limit cycle as ${\bm x}(t)={\bm x}_c + \sqrt{\mu} {\bm x}^{({\rm cm})} (t) +\cdots$, where ${\bm x}_c$ is the fixed point at $\theta=\theta_c$ and ${\bm x}^{({\rm cm})} (t)$ stands for the periodic motion. Then, we we use the lemma \ref{lemma2} and \ref{lemma3} which is explained below, to obtain
\begin{align}
{\cal F}_{2,0}(\theta ) &=\int_0^{T} d t \sum_{\ell} 
\Bigl[  {\partial \sigma \over \partial x_{\ell} }({\bm x}_{c}) + {\sqrt{\mu}}\sum_m \sum_{\ell}\,{\partial^2 \sigma ({\bm x}_{c}) \over \partial x_{\ell} \partial x_{m} } x_{m}^{({\rm cm})} (t)  + \cdots\Bigr]\times
\Bigl[ 
[\tilde{\bm v}_2^{(0)} (t)]_{\ell} + \sqrt{\mu} [\tilde{\bm v}_2^{(1)}(t)]_{\ell} + \cdots \Bigr] \nonumber \\
&=\sqrt{\mu} {\cal F}_{2,0}^{(1)} + \cdots \, , \\
{\cal F}_{2,0}^{(1)} &= \int_0^{2\pi/\omega_0} dt \,\sum_{\ell,m}
{\partial^2 \sigma ({\bm x}_{c}) \over \partial x_{\ell} \partial x_m} x_{m}^{({\rm cm})} (t) 
 [\tilde{\bm v}_2^{(0)}(t)]_{\ell} 
+\sum_{\ell}{\partial \sigma ({\bm x}_c)\over \partial x_{\ell}} 
 [\tilde{\bm v}_2^{(1)} (t)]_{\ell} \, , 
\end{align} 
where we have used $\int_0^{2\pi/\omega_0}dt \tilde{\bm v}_2^{(0)} (t)=0$ and $T={2\pi/\omega_0} + O(\mu)$ (See Eq.(\ref{sm:period})). From the Lemma\ref{lemma2}, we have $\tilde{\bm v}_2^{(0)} (t) = {\cal P} \left( e^{-i\omega_0 \sigma_y t} {\bm \alpha}^{(0)} , {\bm 0}\right)^T$  for the vector ${\bm \alpha}^{(0)}$.  Hence we obtain the behavior:
\begin{align}
[{\rm Var}\sigma]_{\rm D}^{(1)} &= O(\mu^{-1}) =O((\theta - \theta_c)^{-1}) \, . 
\end{align}

\begin{lemma}\label{lemma2}
The left and right eigenvectors for $j=2$ in the original coordinates are expanded as
\begin{align}
\begin{split}
{\bm v}_2 (\theta) &= {\bm v}_2^{(0)} + \sqrt{\mu} {\bm v}_2^{(1)} + \cdots \, , \\
{\bm w}_2 (\theta) &= {\bm w}_2^{(0)} + \sqrt{\mu} {\bm w}_2^{(1)} + \cdots \, .
\end{split}
\end{align}
\end{lemma}

\begin{lemma}\label{lemma3}
The periodic vectors $\tilde{\bm v}_2 (t)$ and $\tilde{\bm w}_2 (t)$ are expanded as
\begin{align}
\begin{split}
\tilde{\bm v}_2 (t) &= \tilde{\bm v}_2^{(0)} (t) + \sqrt{\mu}\tilde{\bm v}_2^{(1)} (t) + \cdots \, , \\
\tilde{\bm w}_2 (t) &= \tilde{\bm w}_2^{(0)} (t)+ \sqrt{\mu} \tilde{\bm w}_2^{(1)} (t) + \cdots \, ,
\end{split}
\end{align}
where 
$\tilde{\bm v}_2^{(0)} (t) = \Phi_{\theta_c} (t,0) {\bm v}_2 (\theta_c)$ and $\tilde{\bm w}_2^{(0)}{}^{\dagger} (t) =  {\bm w}_2^{\dagger} (\theta_c) \Phi_{\theta_c} (0,t)$. 
\end{lemma}

Next, we consider the term $[{\rm Var}\sigma]_{\rm D}^{(2)}$ which is given as follows:
\begin{align}
[{\rm Var}\sigma]_{\rm D}^{(2)} &=\lim_{\tau\to\infty}{1\over \tau} \sum_{\rho}\mu_{\rho}{\rm Re} \int_{0}^{\tau} ds \sum_{\rho} [g_{2,0} (\tau, s)]_{\rho} {\cal A}_{\rho}(\bar{\bm x}(s)) \nonumber \\
&=\lim_{\tau\to\infty}{1\over \tau} \sum_{\rho}\mu_{\rho}{\rm Re} \int_{0}^{\tau} ds 
{ {\cal F}_{2,0} (\theta ) \over \nu_2} (e^{\nu_2 (\tau -s)} - 1 ) \sum_{\ell'} [\tilde{\bm{w}}^\dagger_2(s)]_{\ell'} \nabla_{\rho}^{\ell'} {\cal A}_{\rho}(\bar{\bm x}(s)) \, . 
\end{align}
We introduce the Fourier transform:
\begin{align}
\sum_{\rho}\mu_\rho \sum_{\ell'} [\tilde{\bm{w}}^\dagger_2(s)]_{\ell'} \nabla_{\rho}^{\ell'} {\cal A}_{\rho}(\bar{\bm x}(s)) =\sum_{k'=-\infty}^{\infty} Q_{k'} (\theta )e^{i \omega k' s } \, ,
\end{align}
which leads to 
\begin{align}
[{\rm Var}\sigma]_{\rm D}^{(2)} &={\rm Re} \Bigl[ { {\cal F}_{2,0} (\theta ) \over \nu_2} Q_{0} (\theta ) \Bigr] \, . 
\end{align}
Since $\int_{0}^{T} \d s \tilde{\bm{w}}^{(0)\dagger}_2(s) = 0$, the zeroth-order contribution to $Q_0(\theta)$ vanishes, and therefore $Q_0(\theta)=O(\mu^{1/2})$. Using also ${\cal F}_{2,0} (\theta )=O(\mu^{1/2})$, we obtain
\begin{align}
[{\rm Var}\sigma]_{\rm D}^{(2)} &= O(1) \, .
\end{align}
Thus we conclude
\begin{align}
[{\rm Var}\sigma] & \sim [{\rm Var}\sigma]_{\rm D} \sim [{\rm Var}\sigma]_{\rm D}^{(1)} \propto 
(\theta - \theta_c)^{-1} \, .  
\end{align}

\subsubsection{The case of $\theta  <\theta_c$}
For the regime of $\theta < \theta_c$, the stable solution is not the periodic motion, but the fixed point $\bar{\bm x}$, and hence we use the formula in (\ref{sm:fpformula}) to discuss the fluctuation of entropy production. Note that the stability matrix has finite eigenvalues as in (\ref{sm:hopfeigen}). Eigenvalues in the stable+unstable manifold are $O(1)$, and eigenvalues in the center manifold are near $\pm i \omega_0$ for the parameter near $\theta_c$. This implies that the fluctuation ${\rm Var}\sigma$ never diverges, and hence we conclude $\alpha_- =0_-$, where the symbol $0_-$ denotes that the exponent is nonpositive.

\subsubsection{Proofs for the lemmas 1-3}\label{sm:proofsection}
The proof of the lemma \ref{lemma1} is as follows. 
\begin{proof}
Let $\mu:=\theta-\theta_c>0$, and let
\[
\bar{\bm \zeta}(t)=\bigl(\bar{\bm c}(t)^T,{\bm 0}^T\bigr)^T
\]
be the $T(\mu)$-periodic orbit in the flattened coordinates
${\bm \zeta}=({\bm c}^T,{\bm w}^T)^T$, where
${\bm w}={\bm s}-{\bm h}({\bm c})$.  The corresponding orbit in the original
coordinates is denoted by
\[
\bar{\bm x}(t)=\Psi(\bar{\bm \zeta}(t)).
\]
Here $\Psi$ is the smooth coordinate map from the flattened coordinates to the original coordinates. We define the Jacobian for a later use:
\begin{align}
{\cal J}(t):=\partial_{\bar{\bm \zeta}}\Psi(\bar{\bm \zeta}(t)).
\end{align}

We first clarify the relation between the variational equations in the original and flattened coordinates.  In the original coordinates, the variational equation along $\bar{\bm x}(t)$ is
\begin{align}
\delta\dot{\bm x}
&=
{\mathbb S}_{\theta}(\bar{\bm x}(t))\,\delta{\bm x},
\qquad
{\mathbb S}_{\theta}(\bar{\bm x}(t))
:=
\partial_{\bm x}{\bm F}_{\theta}(\bar{\bm x}(t)).
\end{align}
Equivalently, the same matrix is obtained by first transforming the nonlinear
vector field.  In the flattened coordinates the vector field is given by
\begin{align}
\dot{\bm c}
&=
\hat{\bm g}_{\rm c}({\bm c},{\bm w})
:=
{\bm g}_{\rm c}({\bm c},{\bm h}({\bm c})+{\bm w}),
\\
\dot{\bm w}
&=
\hat{\bm g}_{\rm w}({\bm c},{\bm w})
:=
{\bm g}_{\rm su}({\bm c},{\bm h}({\bm c})+{\bm w})
-
\partial_{\bm c}{\bm h}({\bm c})\,{\bm g}_{\rm c}({\bm c},{\bm h}({\bm c})+{\bm w}).
\end{align}
The stability matrix is computed as
\begin{align}
\tilde{\mathbb S}_{\theta}(t)
=
\partial_{\bm \zeta}
\begin{pmatrix}
\hat{\bm g}_{\rm c}\\
\hat{\bm g}_{\rm w}
\end{pmatrix}
\Bigg|_{({\bm c},{\bm w})=(\bar{\bm c}(t),{\bm 0})}.
\end{align}
Since the center manifold ${\bm s}={\bm h}({\bm c})$ is invariant, one has
\begin{align}
{\bm g}_{\rm su}({\bm c},{\bm h}({\bm c}))
=
\partial_{\bm c}{\bm h}({\bm c})\,{\bm g}_{\rm c}({\bm c},{\bm h}({\bm c})).
\end{align}
Therefore
\begin{align}
\hat{\bm g}_{\rm w}({\bm c},{\bm 0})={\bm 0}
\qquad
\text{for all ${\bm c}$}.
\end{align}
Differentiating this identity with respect to ${\bm c}$ yields
\begin{align}
\partial_{\bm c}\hat{\bm g}_{\rm w}({\bm c},{\bm 0})={\bm 0}.
\end{align}
Consequently, the stability matrix $\tilde{\mathbb S}_{\theta}(t)$ has the block upper triangular form
\begin{align}
\tilde{\mathbb S}_{\theta}(t)
=
\begin{pmatrix}
A_{\rm c}(t) & B(t)\\
{\bm 0} & A_{\rm su}(t)
\end{pmatrix},
\label{eq:block-upper-generator}
\end{align}
where $A_{\rm c}(t)=\partial_{\bm c}\hat{\bm g}_{\rm c}(\bar{\bm c}(t),{\bm 0})$, $
B(t)=\partial_{\bm w}\hat{\bm g}_{\rm c}(\bar{\bm c}(t),{\bm 0})$, and 
$A_{\rm su}(t)=\partial_{\bm w}\hat{\bm g}_{\rm w}(\bar{\bm c}(t),{\bm 0})$.

Let $\tilde{\Phi}_{\theta}(t,0)$ be the fundamental matrix of the flattened
variational equation, and write
\begin{align}
\tilde{\Phi}_{\theta}(t,0)=
\begin{pmatrix}
\tilde{\Phi}_{11}(t,0) & \tilde{\Phi}_{12}(t,0)\\
\tilde{\Phi}_{21}(t,0) & \tilde{\Phi}_{22}(t,0)
\end{pmatrix},
\qquad
\tilde{\Phi}_{\theta}(0,0)=\mathbb I.
\end{align}
From \eqref{eq:block-upper-generator}, the lower-left block satisfies
\begin{align}
\frac{d}{dt}\tilde{\Phi}_{21}(t,0)
=
A_{\rm su}(t)\tilde{\Phi}_{21}(t,0),
\qquad
\tilde{\Phi}_{21}(0,0)={\bm 0}.
\end{align}
By uniqueness of solutions, we have 
\begin{align}
\tilde{\Phi}_{21}(t,0)\equiv {\bm 0}.
\end{align}
Thus the fundamental matrix is written in the following form:
\begin{align}
\tilde{\Phi}_{\theta}(t,0)=
\begin{pmatrix}
\tilde{\Phi}_{\rm c}(t,0) & *\\
{\bm 0} & \tilde{\Phi}_{\rm su}(t,0)
\end{pmatrix},
\end{align}
where
\begin{align}
\frac{d}{dt}\tilde{\Phi}_{\rm c}(t,0)
&=
A_{\rm c}(t)\tilde{\Phi}_{\rm c}(t,0),
\qquad
\tilde{\Phi}_{\rm c}(0,0)=\mathbb I,
\\
\frac{d}{dt}\tilde{\Phi}_{\rm su}(t,0)
&=
A_{\rm su}(t)\tilde{\Phi}_{\rm su}(t,0),
\qquad
\tilde{\Phi}_{\rm su}(0,0)=\mathbb I.
\end{align}
The monodromy matrix in the flattened coordinates is written in the following form
\begin{align}
\tilde{\bm M}_{\theta}
=
\tilde{\Phi}_{\theta}(T,0)
=
\begin{pmatrix}
\tilde{\bm M}^{({\rm c})}_{\theta} & *\\
{\bm 0} & \tilde{\bm M}^{({\rm su})}_{\theta}
\end{pmatrix},
\end{align}
where $\tilde{\bm M}^{({\rm c})}_{\theta}=\tilde{\Phi}_{\rm c}(T,0)$ and $\tilde{\bm M}^{({\rm su})}_{\theta}=\tilde{\Phi}_{\rm su}(T,0)$.
Hence 
$\det(\lambda \mathbb I-\tilde{\bm M}_{\theta}) =\det(\lambda \mathbb I-\tilde{\bm M}^{({\rm c})}_{\theta})\det(\lambda \mathbb I-\tilde{\bm M}^{({\rm su})}_{\theta})$. 
The Floquet multipliers of $\tilde{\bm M}_{\theta}$ are the union of those of the center and su (stable+unstable) blocks. 

We now relate this result to the monodromy matrix in the original coordinates.
Let $\Phi_{\theta}(t,0)$ be the fundamental matrix in the original coordinate. The fundamental matrices in the two coordinates satisfy the relation
\begin{align}
\Phi_{\theta}(t,0) = {\cal J}(t)\,\tilde{\Phi}_{\theta}(t,0)\,{\cal J}^{-1}(0).
\label{eq:fundamental-matrix-gauge-relation}
\end{align}
Since the orbit is periodic, we have $\bar{\bm \zeta}(T)=\bar{\bm \zeta}(0)$ and ${\cal J}(T)={\cal J}(0)$. Therefore the monodromy matrix originally defined by
\begin{align}
{\bm M}_{\theta}
=
{\cal T}\exp\left[\int_0^T dt\,
{\mathbb S}_{\theta}(\bar{\bm x}(t))\right]
=
\Phi_{\theta}(T,0)
\end{align}
is conjugate to the flattened monodromy matrix:
\begin{align}
{\bm M}_{\theta}
=
{\cal J}(0)\,\tilde{\bm M}_{\theta}\,{\cal J}^{-1}(0).
\end{align}
Thus ${\bm M}_{\theta}$ and $\tilde{\bm M}_{\theta}$ have the same Floquet multipliers. The connection term affects the generator $\tilde{\mathbb S}_{\theta}(t)$, but it does not invalidate the above spectral decomposition, because the block triangular structure follows from the invariance of the flattened center manifold ${\bm w}=0$.

We next analyze the center block.  For a generic supercritical Hopf bifurcation,
the reduced dynamics on the two-dimensional center manifold can be written,
after a smooth local change of coordinates, in polar coordinates $(r,\phi)$ as
\begin{align}
\dot r &= \mu r - \gamma r^3 + R_r(r,\mu), \\
\dot \phi &= \omega_0 + R_{\phi}(r,\mu),
\end{align}
where $\gamma>0$. The functions $R_r$ and $R_{\phi}$ are respectively defined as $R_r(r,\mu)=O(\mu^2 r+\mu r^3+r^5)$ and $R_{\phi}(r,\mu)=O(\mu+r^2)$. For $\mu>0$ sufficiently small, there exists a periodic orbit on the center manifold with
\begin{align}
r=\bar r(\mu)=\sqrt{\mu/\gamma}+O(\mu^{3/2}),
\qquad
T(\mu)=\frac{2\pi}{\omega_0}+O(\mu).
\end{align}
Linearizing the radial equation along this orbit gives
\begin{align}
\delta \dot r
&=
\left[
\mu-3\gamma \bar r^2(\mu)+\partial_r R_r(\bar r(\mu),\mu)
\right]\delta r =
\left[-2\mu+O(\mu^2)\right]\delta r,
\end{align}
because $\bar r^2(\mu)=\mu/\gamma+O(\mu^2)$ and
$\partial_r R_r(\bar r(\mu),\mu)=O(\mu^2)$.  Therefore the nontrivial
Floquet exponent in the center directions is
\begin{align}
\nu_2
=
-2\mu+O(\mu^2)
=
-2(\theta-\theta_c)+O((\theta-\theta_c)^2) \, .
\end{align}

The other center Floquet exponent is exactly zero.  Indeed, differentiating the
equation of motion for the periodic orbit on the center manifold gives
\begin{align}
\frac{d}{dt}\dot{\bar{\bm c}}(t)
=
A_{\rm c}(t)\dot{\bar{\bm c}}(t).
\end{align}
Since $\bar{\bm c}(t)$ is $T$-periodic, so is $\dot{\bar{\bm c}}(t)$, and hence
\begin{align}
\tilde{\bm M}^{({\rm c})}_{\theta}\,\dot{\bar{\bm c}}(0)
=
\dot{\bar{\bm c}}(0).
\end{align}
Thus one center Floquet multiplier equals $1$, corresponding to the Floquet
exponent
\begin{align}
\nu_1=0.
\end{align}

Finally, we consider the su (stable+unstable) block.  As $\mu\to 0^+$, the periodic orbit
shrinks to the Hopf fixed point, and therefore
\begin{align}
A_{\rm su}(t)=A_{{\rm su},0}+O(\sqrt{\mu}),
\qquad
A_{{\rm su},0}:=
\tilde{\mathbb S}^{({\rm su})}_{\theta_c}({\bm 0}),
\end{align}
uniformly for $t\in[0,T(\mu)]$.  The eigenvalues of $A_{{\rm su},0}$ are the
eigenvalues of the linearization at the Hopf point restricted to
$E^{\rm s}\oplus E^{\rm u}$, and hence all of them have nonzero real parts.
Consequently,
\begin{align}
\tilde{\bm M}^{({\rm su})}_{\theta}
=
e^{T_0 A_{{\rm su},0}}+O(\sqrt{\mu}),
\qquad
T_0:=\frac{2\pi}{\omega_0}.
\end{align}
Choosing the branches of the logarithm continuously from $\mu=0$, standard
perturbation theory for periodic linear systems gives, for $j=3,\dots,n$,
\begin{align}
\nu_j(\theta)=\lambda_j(\theta_c)+o(1),
\end{align}
where $\lambda_j(\theta_c)$ are the eigenvalues of $A_{{\rm su},0}$.  In
particular, these $n-2$ exponents remain of order $O(1)$.

Collecting the center and su spectra, we conclude that for
$\theta>\theta_c$ sufficiently close to $\theta_c$, the Floquet exponents
consist of one zero exponent, one exponent of order
$O(\theta-\theta_c)$, and $n-2$ exponents of order $O(1)$.
This completes the proof.
\end{proof}

\vspace*{1.5cm}

The proof of the lemma \ref{lemma2} is as follows. 
\begin{proof}
We consider the flattened ${\bm \zeta}$ coordinate. The monodromy matrix in the ${\bm \zeta}$ coordinate has the following structure 
\begin{align}
\tilde{\bm M}_{\theta} &=
\left( 
\begin{array}{ll}
\tilde{\bm M}_{\rm c} \, , & \tilde{\bm M}_{\rm csu} \\
{\bm 0} \, , & \tilde{\bm M}_{\rm su}
\end{array}
\right) = {\cal J}^{-1}(0) {\bm M}_{\theta} {\cal J}(0) \, . 
\end{align}
Let us denote the right eigenvector of $\tilde{\bm M}_{\theta}$ as ${\bm \zeta}_2$, which is related to the full space vector via ${\bm v}_2 (\theta) ={\cal J}(0) {\bm \zeta}_2$. We decompose the vector ${\bm \zeta}_2$ as
\begin{align}
{\bm \zeta}_2 &= 
\left( 
\begin{array}{l}
{\bm \alpha} \\ {\bm \beta}
\end{array}
\right) .
\end{align}
The eigenvalue equation reads
\begin{align}
\begin{split}
\tilde{\bm M}_{\rm c} {\bm \alpha} + \tilde{\bm M}_{\rm csu} {\bm \beta} &= e^{\nu_2 T} {\bm \alpha} \, , \\
\tilde{\bm M}_{\rm su} {\bm \beta} &= e^{\nu_2 T} {\bm \beta} \, .
\end{split}
\end{align}
Since $\nu_2 = {O}(\mu) \ll 1$, we have $e^{\nu_2 T} \sim 1$. The second equation implies that ${\bm \beta}$ must be a zero vector, i.e., ${\bm \beta}={\bm 0}$, because the eigenvalues of $\tilde{\bm M}_{\rm su}$ are $O(1)$ and bounded away from $1$. The eigenvalue equation is then reduced to $\tilde{\bm M}_{\rm c} {\bm \alpha}= e^{\nu_2 T} {\bm \alpha}$. Because $\tilde{\bm M}_{\rm c}$ is expanded as $\tilde{\bm M}_{\rm c}\sim {\mathbb I} + {O(\mu)}$, we have
${\bm \alpha}= {\bm \alpha}^{(0)} + {O}(\mu) {\bm \alpha}^{(1)}$. 

To obtain the eigenvector in the original coordinates, we multiply it by the transformation Jacobian:
\begin{align}
{\bm v}_2 (\theta )&= {\cal J}(0) \left( \begin{array}{l} {\bm \alpha} \\ {\bm 0} \end{array} \right)
= {\cal P} \left( \begin{array}{cc} {\mathbb I} & {\bm 0} \\\partial_{\bm c} {\bm h}({\bm c}(0)) & {\mathbb I} \end{array} \right) 
 \left( \begin{array}{l}
{\bm \alpha}^{(0)} + {O}(\mu) {\bm \alpha}^{(1)} \\ {\bm 0}
\end{array} \right) .
\end{align}
On the limit cycle of radius $r \sim \sqrt{\mu}$, the derivative of the center manifold scales as $\partial_{\bm c} {\bm h}({\bm c}) = {O}(|{\bm c}|) = {O}(\sqrt{\mu})$. Hence, the nonlinear transformation introduces an $O(\sqrt{\mu})$ correction to the full-space eigenvector:
\begin{align}
{\bm v}_2 (\theta )&= {\cal P} \left( \begin{array}{l} {\bm \alpha}^{(0)} \\ {\bm 0} \end{array} \right) + \sqrt{\mu} {\bm v}_2^{(1)} +\cdots= {\bm v}_2^{(0)} + \sqrt{\mu} {\bm v}_2^{(1)} +\cdots \, .
\end{align}
A similar argument applies to the left eigenvector, yielding ${\bm w}_2 (\theta) = {\bm w}_2^{(0)} + \sqrt{\mu} {\bm w}_2^{(1)} + \cdots $.
\end{proof}

\vspace*{1.5cm}

The proof of the lemma \ref{lemma3} is as follows. 
\begin{proof}
Note that we consider $t\neq 0$ and $t\neq T$. We map the state transition matrix back to the original coordinates using the Jacobian of the transformation evaluated along the limit cycle, ${\cal J}(t)$. We write $\tilde{\Phi}_{\theta} (t,0)={\cal J}^{-1}(t)\Phi_{\theta} (t,0) {\cal J}(0)$. Then we have
\begin{align}
\tilde{\bm v}_2 (t) &= {\Phi}_{\theta} (t,0) {\bm v}_2 (\theta ) =
{\cal J}(t) \tilde{\Phi}_{\theta} (t,0) {\cal J}^{-1}(0) {\bm v}_2 (\theta ) \nonumber \\
&={\cal J}(t)\tilde{\Phi}_{\theta} (t,0) \left( \begin{array}{l}{\bm \alpha} \\ {\bm 0} \end{array}\right) 
= {\cal J}(t)
\left( \begin{array}{l}
\hat{T}\exp\Bigl[\int_0^{t} ds \tilde{\mathbb S}_{\theta}^{({\rm c})} (s) \Bigr] {\bm \alpha} \\ {\bm 0} \end{array}\right) \, .
\end{align}
By expanding the exponential as in standard perturbation theory, we note:
\begin{align}
\hat{T}\exp\Bigl[\int_0^{t} ds \tilde{\mathbb S}_{\theta}^{({\rm c})} (s) \Bigr]
&= \tilde{\Phi}_{\theta_c} (t,0) + {O}(\sqrt{\mu}) \, .
\end{align}
Furthermore, since $\partial_{\bm c}{\bm h}({\bm c}(t)) = {O}(\sqrt{\mu})$ along the limit cycle, the transformation matrix scales as ${\cal J}(t) = {\cal P} + {O}(\sqrt{\mu})$. This implies 
\begin{align}
\tilde{\bm v}_2 (t) &= \left[ {\cal P} + {O}(\sqrt{\mu}) \right] \left( \begin{array}{l} 
\tilde{\Phi}_{\theta_c} (t,0) {\bm \alpha}^{(0)} + {O}(\sqrt{\mu}) \\
{\bm 0}
\end{array} \right)  \nonumber \\
&= \tilde{\bm v}^{(0)}_2 (t) + \sqrt{\mu} \tilde{\bm v}_2^{(1)} \, . 
\end{align}
The same logic as this derivation is applied to obtain the relation for $\tilde{\bm w}_2 (t)$.
\end{proof}

\section{Parametric response of entropy production rate}\label{res_entropyprod}
As explained in Sec.\ref{classification_of_entropyfluctuation}, we consider the case where the concentration $a_{\ell}$ in the chemostat so that one can control the amplitude of the parameter $\theta$. The parameter $\theta$ is a control parameter to govern the bifurcations. In this setup, we discuss {\it general} diverging behavior of the response of the average entropy production against the change of parameter $\theta$ by defining the exponents. 
We define the parametric response as
\begin{align}
\partial_{\theta}\sigma &= {\partial \over \partial \theta }
\lim_{\tau\to}{1\over \tau} \int_0^{\tau} dt \sigma (t) \, .
\end{align}
We then look at the diverging behavior with the exponents $\beta_{\pm}$:
\begin{align}
{\partial_{\theta} \sigma }& \propto
\left\{
\begin{array}{ll}
(\theta - \theta_c)^{-\beta_+} & \theta > \theta_c \\
(\theta_c - \theta )^{-\beta_-} & \theta < \theta_c \\
\end{array} 
\right.
\end{align}

\subsection{Fixed-point cases}\label{res_fixpoint}
In the case of fixed points, we write the entropy production as $\sigma (\theta , \bar{\bm x})$, where $\bar{\bm x}$ is a $\theta$-dependent fixed point. Hence one can write the parametric response as 
\begin{align}
{\partial_{\theta} \sigma} &=
{\partial \sigma (\theta, \bar{\bm x})\over \partial \theta } + \sum_{\ell}{\partial \sigma (\theta, \bar{\bm x})\over \partial \bar{x}_{\ell}} {\partial \bar{x}_{\ell} \over \partial \theta} \, .
\nonumber \\
&= {\partial \sigma (\theta, \bar{\bm x})\over \partial \theta } + [{\rm grad}\sigma]^{\top} {\cal J} (\bar{\bm c}) {\partial \bar{\bm \zeta} \over \partial \theta} \, . 
\end{align} 
From the property 3 in Sec.\ref{sm:pro}, $[{\rm grad}\sigma]^{\top}$ is finite,  and hence the singular behavior can appear from the term ${\partial \bar{\bm \zeta}/ \partial \theta}$.

\noindent\\
In the case of pitchfork bifurcation, we have the solution $\bar{c}=\pm\sqrt{\mu}$ for $\mu:=\theta -\theta_c>0$, and $0$ for $\mu<0$, and hence we have $(\beta_+,\beta_-)=(1/2,0_-)$. 

\noindent\\
In the case of transcritical bifurcation, we have the solution $\bar{c}=\mu$ for $\mu >0$, and $0$ for $\mu<0$, and hence we have $(\beta_+,\beta_-)=(0_-,0_-)$.  

\noindent\\
In the case of saddle-node bifurcation, we have the solution $\bar{c}=\pm \sqrt{\mu}$ for $\mu>0$, and hence we have $(\beta_+,\beta_-)=(1/2,\varnothing)$.  

\subsection{Hopf bifurcation}\label{res_hopf}
In the case of a supercritical Hopf bifurcation, the normal form yields the solution
\begin{align}
\bar{\bm{c}}(t) &= 
\left\{
\begin{array}{ll}
\bm{0}, & \quad \text{for } \mu <0 ,\\
\sqrt{\mu}
\begin{pmatrix}
\cos(\omega t) \\
\sin(\omega t)
\end{pmatrix}
, &\quad \text{for } \mu >0 \, . 
\end{array}
\right.
\end{align}
For $\mu >0$, the solution is periodic with period $T = 2\pi/\omega$. Accordingly, we consider
\begin{align}
\partial_{\theta}\sigma
&= \frac{\partial}{\partial \theta} \sum_{\rho} \mu_{\rho}
\frac{1}{T} \int_{0}^{T}
\left( J_{+\rho}(\bar{\bm{x}}(t)) - J_{-\rho}(\bar{\bm{x}}(t)) \right) dt.
\end{align}
By the law of mass action, $J_{\pm \rho}$ are polynomials in the components of $\bar{\bm{c}}(t)$. Therefore, they can be expressed as polynomials of $\sqrt{\mu} \cos(\omega t)$, $\sqrt{\mu} \sin(\omega t)$, and constants.
Using the identity
\begin{equation}
\frac{1}{2\pi} \int_{0}^{2\pi} (\cos x)^m (\sin x)^n \, dx = 0,
\quad \text{if either $m$ or $n$ is odd},
\end{equation}
it follows that the term 
$\frac{1}{T} \int_{0}^{T}
\left( J_{+\rho}(\bm{x}(t)) - J_{-\rho}(\bm{x}(t)) \right) dt
$
contains only integer powers of $\mu$, and no half-integer powers. Consequently, we obtain $(\beta_-, \beta_+) = (0_-, 0_-)$.

\section{General inequality between $\alpha$ and $\beta$}
In this section, we derive the general inequality between $\alpha$ and $\beta$, valid for any kinds of bifurcation in the chemical reactions. We use the environment entropy (\ref{entropybath}). One can easily find the following identity:
\begin{align}
{\partial_{\theta}} \langle \Sigma_{\rm e} (\Gamma) \rangle - \langle \partial_{\theta}\Sigma_{\rm e} (\Gamma) \rangle 
&=
\int d\Gamma \, {\cal P}(\Gamma ) {\partial \ln {\cal P}(\Gamma) \over \partial \theta} \Sigma_{\rm e} (\Gamma ) \, .
\end{align}
The Cauchy-Schwarz inequality leads to the Cramer-Rao type inequality
\begin{align}
\Bigl[ {\partial_{\theta}} \langle \Sigma_{\rm e} (\Gamma) \rangle - \langle \partial_{\theta}\Sigma_{\rm e} (\Gamma) \rangle  \Bigr]^2 &\le I(\theta , \tau)\, {\rm Var}\Sigma_{\rm e} \, , 
\end{align}
where $I(\theta, \tau)$ and ${\rm Var}\Sigma_{\rm e}$ are respectively defined as
\begin{align}
I(\theta, \tau) &= \Bigl< -{\partial^2 \over \partial \theta^2} \ln {\cal P} (\Gamma)\Bigr> \, , \\
{\rm Var}\Sigma_{\rm e} &= \Bigl< \Bigl[ \sum_{\rho} \mu_{\rho} Z_{\rho} (\Gamma)\Bigr]^2 \Bigr> -\Bigl< \sum_{\rho} \mu_{\rho} Z_{\rho} (\Gamma) \Bigr>^2 \, . 
\end{align}
The quantity $I(\theta , \tau)$ is the Fisher information computed through the path trajectory. We define the following scaled Fisher information 
\begin{align}
i(\theta) &:= \lim_{\tau \to \infty}{1\over \tau} \lim_{\Omega \to \infty} {1\over \Omega} I(\theta , \tau) \, . 
\end{align}
We obtain the following inequality
\begin{align}
\Bigl[ {\partial_{\theta} \sigma} -
\lim_{\tau \to \infty}{1\over \tau} \lim_{\Omega \to \infty} {1\over \Omega}  \langle \partial_{\theta}\Sigma_{\rm e} (\Gamma) \rangle  \Bigr]^2 &\le i(\theta )\, {\rm Var}\sigma \, . \label{sm:basedinequality}
\end{align}
We first discuss the property of the second part in the left hand side. Using the exact micro-macro correspondence proven by Kurtz \cite{Kurtz10.1063/1.1678692,Kurtz_1971}, we have  
\begin{align}
\lim_{\tau \to \infty}{1\over \tau} \lim_{\Omega \to \infty} {1\over \Omega}  \langle \partial_{\theta}\Sigma_{\rm e} (\Gamma) \rangle &= \sum_{\rho} ({\partial_{\theta} \mu_{\rho} }) \lim_{\tau\to\infty} {1\over \tau} \int_0^{\tau} dt \, J_{\rho} ({\bm x}(t)) \, , 
\end{align}
From the property 1 in Sec.\ref{sm:pro}, the long time average of $J_{\rho}$ must be finite, and hence this term does not show the divergence. 

We next consider the general properties of the scaled Fisher information as follows. We first observe the following expression for the Fisher information.
\begin{align}
I(\theta , \tau) &=-\int d\Gamma \, \Bigl[{\partial^2 \over \partial \theta^2} \ln {\cal P} (\Gamma ) \Bigr] {\cal P} (\Gamma ) \nonumber \\
&=-\sum_{\rho}\left[ \sum_{{\bm n} \neq {\bm n}'} \left({\partial^2 \over \partial \theta^2} \ln W_{{\bm n},{\bm n}'}^{\rho} \right) W_{{\bm n},{\bm n}'}^{\rho} P({\bm n}', t) -\sum_{{\bm n} \neq {\bm n}'} \left({\partial^2 \over \partial \theta^2} W_{{\bm n},{\bm n}'}^{\rho} \right) P({\bm n}', t)  \right] \nonumber \\
&= \sum_{\rho}\sum_{{\bm n} \neq {\bm n}'} \left({\partial \over \partial \theta} \ln W_{{\bm n},{\bm n}'}^{\rho} \right)^2 W_{{\bm n},{\bm n}'}^{\rho} P({\bm n}', t)  \, .
\end{align}
Now we use the detailed expression of the transition rate (\ref{sm:wrhonn}) to obtain
\begin{align}
{\partial \over \partial \theta}
\ln W^{\rho}_{\bm{n}, \bm{n} \mp \nabla_{\rho}} 
&= {\partial \over \partial \theta} \sum_{\ell' \in \mathcal{S}_c} { \nabla_{\pm\rho}^{\ell'}} \ln a_{\ell'} =\sum_{\ell' \in \mathcal{S}_c} { \nabla_{\pm\rho}^{\ell'}} {\hat{a}_{\ell'}\over a_{\ell'} }
= \sum_{\ell' \in \mathcal{S}_c} { \nabla_{\pm\rho}^{\ell'} \over \theta + a_{\ell'}^{(0)}/\hat{a}_{\ell'}} \, ,
\end{align}
where we use the definition of $\theta$ in (\ref{sm:defoftheta}). The scaled Fisher information is thus given as follows
\begin{align}
i(\theta) &= \lim_{\tau\to\infty} {1\over \tau} \int_{0}^{\tau} \sum_{\rho} \left[ \left( \sum_{\ell' \in \mathcal{S}_c} { \nabla_{+\rho}^{\ell'} \over \theta + a_{\ell'}^{(0)}/\hat{a}_{\ell'}} \right)^2  J_{+\rho } ({\bm x}(t)) + \left( \sum_{\ell' \in \mathcal{S}_c} { \nabla_{-\rho}^{\ell'} \over \theta + a_{\ell'}^{(0)}/\hat{a}_{\ell'}} \right)^2 
J_{-\rho } ({\bm x}(t)) \right] \nonumber \\
&\le c_{\rm max} 
\lim_{\tau\to\infty} {1\over \tau} \int_{0}^{\tau} \sum_{\rho} \left[J_{+\rho } ({\bm x}(t)) 
+ J_{-\rho } ({\bm x}(t)) \right] \, , \label{sm:ithetaexpression}
\end{align} 
where 
\begin{align}
c_{\rm max} &:=\max_{\rho} \left\{ \left( \sum_{\ell' \in \mathcal{S}_c} { \nabla_{+\rho}^{\ell'} \over \theta + a_{\ell'}^{(0)}/\hat{a}_{\ell'}} \right)^2 , \left( \sum_{\ell' \in \mathcal{S}_c} { \nabla_{-\rho}^{\ell'} \over \theta + a_{\ell'}^{(0)}/\hat{a}_{\ell'}} \right)^2 \right\} \, .
\end{align}
The inequality (\ref{sm:ithetaexpression}) indicates that the scaled Fisher information is related to the activity per unit volume in the thermodynamic limit. From the property 1 in Sec.\ref{sm:pro}, the scaled Fisher information generically does not show the divergence. 

As a result, the inequality (\ref{sm:basedinequality}) immediately leads to the following inequality on the exponents
\begin{align}
\alpha &\ge 2 \beta \, ,
\end{align}
for $\alpha \ge 0$ and $\beta \ge 0$. 

\end{widetext}

\end{document}

%% file: pitchfork.tex
\begin{tikzpicture}[scale=1, baseline=(current bounding box.center)]
  \draw[->, black!100] (-1.5,0) -- (1.5,0) node[below] {$\theta$};
  \draw[->, black!100] (0,-1.6) -- (0,1.8) node[above] {$x$};

  \node[below left] at (0,0) {$\theta_c$};

  \draw[thick, blue!70!black] (-1.5,0) -- (0,0);
  \draw[dashed, thick, red!70!black] (0,0) -- (1.3,0);

  \draw[thick, blue!70!black, domain=0:1.3, samples=100]
    plot ({\x},{sqrt(\x)});
  \draw[thick, blue!70!black, domain=0:1.3, samples=100]
    plot ({\x},{-sqrt(\x)});
\end{tikzpicture}

%% file: transcritical.tex
\begin{tikzpicture}[scale=1, baseline=(current bounding box.center)]
  \draw[->, black!100] (-1.5,0) -- (1.5,0) node[below] {$\theta$};
  \draw[->, black!100] (0,-1.6) -- (0,1.8) node[above] {$x$};

  \draw[black!0] (0,-1.8) -- (0,-1.6);
  \node[below left] at (0,0) {$\theta_c$};

  \draw[thick, blue!70!black] (-1.3,0) -- (0,0);
  \draw[dashed, thick, red!70!black] (0,0) -- (1.3,0);

  \draw[dashed, thick, red!70!black] (-1.3,-1.3) -- (0,0);
  \draw[thick, blue!70!black] (0,0) -- (1.3,1.3);
\end{tikzpicture}

%% file: saddlenode.tex
\begin{tikzpicture}[scale=1, baseline=(current bounding box.center)]
  \draw[->, black!100] (-0.3,0) -- (1.5,0) node[below] {$\theta$};
  \draw[->, black!100] (0,-1.8) -- (0,1.6) node[above] {$x$};

  \node[below left] at (0,0) {$\theta_c$};

  \draw[dashed, thick, red!70!black, domain=0:1.3, samples=100]
    plot ({\x},{-sqrt(\x)});
  \draw[thick, blue!70!black, domain=0:1.3, samples=100]
    plot ({\x},{sqrt(\x)});
\end{tikzpicture}

%% file: hopf.tex
\begin{tikzpicture}[scale=1, baseline=(current bounding box.center)]

  \tdplotsetmaincoords{70}{15}

  \begin{scope}[tdplot_main_coords]

    \draw[->, black!100] (-1.5,0,0) -- (1.8,0,0) node[right] {$\theta$};
    \draw[->, black!100] (0,1.8,0) -- (0,-1.8,0);   
    \draw[->, black!100] (0,0,-1.8) -- (0,0,1.8) node[left] {$x$};

    \node[black!100] at (0,-1.7,-0.2) {$y$};

    \node[below left] at (0,0,0) {$\theta_c$};

    \draw[thick, blue!70!black] (-1.5,0,0) -- (0,0,0);
    \draw[dashed, thick, red!70!black] (0,0,0) -- (1.5,0,0);

    \foreach \t in {0.4,0.9,1.4} {
      \pgfmathsetmacro\r{sqrt(\t)}
      \draw[thick, blue!70!black]
        plot[domain=0:360, samples=80]
        ({\t},{-\r*cos(\x)},{\r*sin(\x)});
    }

  \end{scope}
\end{tikzpicture}

%% file: div_twosided.tex
\begin{tikzpicture}[scale=1, baseline=(current bounding box.center)]
  \draw[->, black!100] (-1.5,0) -- (1.5,0) node[below] {$\theta$};
  \draw[->, black!100] (0,0) -- (0, 2.2) node[above]  {$\text{Var}\sigma$};

  \node[black!100, below] at (0,0) {$\theta_c$};

  \draw[thick, blue!70!black, domain=-1.3:-0.12, samples=80]
    plot ({\x},{0.25/(-\x)});

  \draw[thick, blue!70!black, domain=0.12:1.3, samples=80]
    plot ({\x},{0.25/(\x)});

\end{tikzpicture}

%% file: div_saddle_node.tex
\begin{tikzpicture}[scale=1, baseline=(current bounding box.center)]
  \draw[->, black!100] (-0.3,0) -- (1.5,0) node[below] {$\theta$};
  \draw[->, black!100] (0,0) -- (0, 2.2) node[above]  {$\mathrm{Var}\,\sigma$};

  \node[black!100, below] at (0,0) {$\theta_c$};


  \draw[thick, blue!70!black, domain=0.12:1.3, samples=80]
    plot ({\x},{0.25/(\x)});
\end{tikzpicture}

%% file: div_lambdatype.tex
\begin{tikzpicture}[scale=1, baseline=(current bounding box.center)]
  \draw[->, black!100] (-1.5,0) -- (1.5,0) node[below] {$\theta$};
  \draw[->, black!100] (0,0) -- (0, 2.2) node[above]  {$\mathrm{Var}\,\sigma$};

  \node[black!100, below] at (0,0) {$\theta_c$};

  \draw[thick, blue!70!black, domain=-1.3:-0.0, samples=100]
    plot ({\x},{0.35 + 0.4*exp(4*\x)});

  \draw[thick, blue!70!black, domain=0.12:1.3, samples=80]
    plot ({\x},{0.25/(\x)});
\end{tikzpicture}